\documentclass{article}
\usepackage{paper}
\pdfoutput=1

\def\tildeA{\widetilde{A}}

\def\outalg{\text{\sf\small rls-out}}

\begin{document}

\title{Revisiting Differentially Private Regression:\\
Lessons From Learning Theory and their Consequences}
\author{
Xi Wu$^{\dagger}$
\ \ 
Matthew Fredrikson$^\ddagger$
\ \ 
Wentao Wu$^{\ast}$
\ \ 
Somesh Jha$^{\dagger}$
\ \ 
Jeffrey F. Naughton$^{\dagger}$
\ \ 
\\
$^{\dagger}$University of Wisconsin-Madison,
$^{\ddagger}$Carnegie Mellon University\\
$^{\ast}$Microsoft Research\\
\{xiwu, wentaowu, jha, naughton\}@cs.wisc.edu, mfredrik@cs.cmu.edu, 
}
\maketitle

\begin{abstract}
  Private regression has received attention from both database and security
  communities. Recent work by Fredrikson et al. (USENIX Security 2014)
  analyzed the functional mechanism (Zhang et al. VLDB 2012) for training
  linear regression models over medical data. Unfortunately, they found that
  model accuracy is already unacceptable with differential privacy
  when $\varepsilon = 5$. We address this issue, presenting
  an explicit connection between differential privacy and stable
  learning theory through which a substantially better privacy/utility tradeoff
  can be obtained. Perhaps more importantly, our theory reveals that the most
  basic mechanism in differential privacy, \emph{output perturbation},
  can be used to obtain a better tradeoff for \emph{all} convex-Lipschitz-bounded
  learning tasks. Since output perturbation is simple to implement,
  it means that our approach is potentially widely applicable in practice.
  We go on to apply it on the same medical data as used by Fredrikson
  et al. Encouragingly, we achieve accurate models even for
  $\varepsilon = 0.1$. In the last part of this paper, we study the impact of
  our improved differentially private mechanisms on \emph{model inversion attacks},
  a privacy attack introduced by Fredrikson et al. We observe that the improved
  tradeoff makes the resulting differentially private model
  \emph{more susceptible} to inversion attacks. We analyze this phenomenon
  formally.
\end{abstract}


\section{Introduction}
\label{sec:intro}
Differential-private data analytics has received considerable
attention from the data management community~\cite{JYC15, LQSC12,
XZXYYW13, ZXYZW13}. This paper focuses on regression models, which are
widely used to extract valuable data patterns in ``sensitive''
domains.  For example, in personalized medicine, linear regression
models are commonly used to predict medication dosages~\cite{IWPC09}
and other useful features~\cite{CPPM12, RAFDA11}. For such scenarios,
individuals' privacy has become a major concern and learning with
differential privacy (DP) is particularly desirable.

Differentially private regression, and more generally private convex learning,
has been intensively studied in recent years by researchers from the data
management, security, and theory communities~\cite{BST14, CMS11, DJW13,
FLJLPR14, JT13, ZZXYW12}. One notable contribution is the
\emph{functional mechanism}, proposed by Zhang et al.~\cite{ZZXYW12},
which is a practical mechanism for training differentially-private
regression models.  Several research groups~\cite{APW15, WSW15, WYZ12}
have adopted the functional mechanism as a basic building block in
their study and initial empirical results are promising.

Our starting point is recent work by Fredrikson et al.~\cite{FLJLPR14}, which
used the functional mechanism to train differentially-private linear
regression models to predict doses of a medicine called \emph{warfarin} from
patients' genomic traits. Unfortunately, they found that the model accuracy was
unacceptable with $\varepsilon$-DP --- even for $\varepsilon$ as high as~$5$.
One of the main goals of this paper is to develop a suitable theory which
allows mechanisms based on simple techniques such as \emph{output perturbation}
to obtain accurate, differentially-private models at reasonable private levels
(e.g., $\varepsilon=0.1$).

To this end, we start by presenting an explicit and precise connection between
differential privacy and stability theory in computational learning.
Specifically, in the setting of learning, we show that differential privacy is
essentially another stability definition, yet it is so strong that it implies
previous stability definitions in the learning literature. Combining this
observation with some machinery from stable learning theory, we give an
analysis showing that simple mechanisms such as
output perturbation can learn \emph{every} convex-Lipschitz-bounded problem
with \emph{strong} differential privacy guarantees.

Our analysis has several advantages over previous work.  First, it
{\em relaxes} the technical conditions required by Chaudhuri et
al.~\cite{CMS11}.  In particular, we do not require the loss function
to be smooth or differentiable. Second, our analysis reveals that
output perturbation can be used to obtain a much better
privacy/utility tradeoff than the functional mechanism. Under the same
technical conditions, we achieve the \emph{same tradeoff} between
differential privacy and generalization error as that recently proved
by Bassily et al.~\cite{BST14}, while avoiding use of the exponential
mechanism~\cite{MT07} and the sophisticated sampling sub-procedure
used by Bassily et al.~\cite{BST14}. This makes our approach
widely-applicable in practical settings.

We then apply our theory to {\em linear regression models}, for which we
notice that the functional mechanism is the state-of-the-art approach being
applied in the literature~\cite{APW15, FLJLPR14, WSW15}. Our mechanisms are
\emph{regularized} versions of least squares optimization. In contrast to the
functional mechanism, where regularization is a heuristic, regularization is
crucial for our theoretical guarantee.

Regularization adds additional parameters that need to be picked
carefully in order to maintain DP.  A standard approach used in
related contexts is to employ a private parameter tuning algorithm~\cite{CMS11}
to select a set of parameters that depends on the training data.
We have implemented this approach to parameter tuning
and refer to the mechanisms using it as the {\em privately-tuned mechanisms}
in our empirical study. We also consider a different approach,
based on our proof of generalization error,
that picks the parameters in a {\em data independent} manner.

We evaluate both approaches using the same data and experimental
methodology used by Fredrikson et al.~\cite{FLJLPR14}. The results are
encouraging --- both approaches produce substantially more accurate
models than the functional mechanism of Zhang et al.~\cite{ZZXYW12}.
Specifically, two tuned mechanisms obtain accurate models with $\varepsilon$-DP
for $\varepsilon = 0.3$.
Interestingly, the data-independent mechanism significantly outperforms
the tuned ones for small $\varepsilon$ and obtains accurate models
even for $\varepsilon = 0.1$! Finally, using the same simulation approach described
by Fredrikson et al.~\cite{FLJLPR14}, we note that as
compared to the functional mechanism, our methods induce significantly
smaller risk of mortality, bleeding and stroke, especially when
compared at small $\varepsilon$ settings.

Our results described thus far are positive, and follow in the
tradition of a large body of research on differential privacy, which
seeks to improve utiltiy/privacy tradeoffs for various problems of
interest. However, Fredrikson et al.~\cite{FLJLPR14} did not only
consider differential privacy, they also considered an attack they
term \emph{model inversion}.  As a simple example, consider a machine
learning model $w$ that takes features $x_1, \dots, x_d$ and produces
a prediction $y$.  A model-inversion (MI) attack takes input
$x_1, \dots, x_{d-1}$ and a value $y'$ that is related to $y$, and
tries to predict $x_d$ (thus ``inverting the model''). For example,
in~\cite{FLJLPR14}, the authors consider the case where $x_d$ is a
genetic marker, $y$ is the warfarin dosage, and $x_1,\dots, x_{d-1}$
are general background information such as height and weight.  They
used an MI attack to predict an individual's \emph{genetic markers}
based on his or her \emph{warfarin dosage}, thus violating that individual's
privacy.

In the final part of our paper, we consider the impact of our improved
privacy-utility tradeoff on MI attacks. Here our results are less
positive: in improving the privacy-utility tradeoff, we have increased
the effectiveness of MI attacks. While unintended, upon deeper reflection
this is not surprising: simply put, the improved privacy-utility tradeoff
results in less noise added, and less noise added means the model is easier to invert.
We formalize this discussion and prove that this phenomenon is quite general,
and not an artifact of our approach or this specific learning task.

What this means for privacy is an open question. If one ``only cares''
about differential privacy, then the increased susceptibility to MI
attacks is irrelevant. However, if one believes that MI attacks are
significant (and anecdotal evidence suggests that some medical
professionals are concerned about MI attacks), then the fact that
improved differential privacy can mean worse MI exposure warrants
further study. In this direction, our work indeed extends a long line of work
that discusses the interaction of DP and {\em attribute privacy}
~\cite{KM11, LO11, RAWHPS10}, and gives a realistic application
where misconceptions about DP can lead to unwanted disclosure.
It is our hope that our work might highlight this issue
and stimulate more discussion.

Our technical contributions can be summarized as follows:
\begin{itemize}
\item We give an explicit connection between differential privacy and
  stability theory in machine learning.  We show that differential privacy is
  essentially a strong stability notion that implies well-known previous
  notions.

\item Moreover, combining this connection with some machinery from stability theory,
  we prove that the simple output perturbation mechanism can learn \emph{every}
  convex-Lipschitz-bounded learning problems with strong differential privacy.
  Our analysis relaxes the technical conditions required by Chaudhuri et al.~\cite{CMS11},
  achieves the same tradeoff between DP and generalization error
  as proved by Bassily et al.~\cite{BST14}, and is much simpler than both analyses.

\item Since output perturbation is one of the simplest mechanisms to
  implement, it means that our method is widely applicable in practice.
  We go on to apply the theory to linear regression. We present and analyze
  regularized variants of linear regression, and give a detailed description
  on how to privately select parameters for regularization.

\item We perform a re-evaluation, using the same medical data set and
  experimental methodology as previous work~\cite{FLJLPR14}, comparing the
  functional mechanism with our mechanisms to train linear regression models.
  Encouragingly, we observe a substantially better tradeoff between
  differential privacy and utility.

\item We perform a re-evaluation of model inversion in the same medical data
  analysis setting, and demonstrate that as we improve differentially-private
  mechanisms, MI attacks become more problematic. We provide a theoretical
  explanation of this phenomenon.
\end{itemize}

The rest of the paper is organized as follows: In Section~\ref{sec:prelim} we
present background knowledge that will facilitate our discussion later.
Next we present the connection between differential privacy and
stability theory in Section~\ref{sec:dp-and-stability},
and discuss the applications in Section~\ref{sec:applications}.
In Section~\ref{sec:evaluation}, we compare output perturbation and the
functional mechanism with respect to the privacy-utility or model-inversion
efficacy tradeoff. Finally, Section~\ref{sec:related} discusses related work
and in Section~\ref{sec:conclusion}, we provide concluding remarks.


\section{Preliminaries}
\label{sec:prelim}
In this section we present some background in three different areas: learning
theory, convex optimization and differential privacy. We will only cover some
basics of these areas to facilitate our discussion later. Readers are referred
to~\cite{BV04},~\cite{DR14} and~\cite{UML} for an in-depth introduction to
these three topics, respectively.

\subsection{Learning Theory}
\label{preliminaries:learning-theory}
Let $X$ be a feature space, $Y$ be an output space, and $Z = X \times Y$ be a
sample space. For example, $Y$ is a set of labels for classification, or an
interval in $\Real$ for regression. Let $\cal H$ be a hypothesis space
and $\ell$ be an \emph{instance-wise} loss function, such that on an input
hypothesis $w \in {\cal H}$ and a sample $z \in Z$, it gives a loss $\ell(w, z)$.

Let $\cal D$ be a distribution over $Z$, the goal of learning is to find a
hypothesis $w \in {\cal H}$ so as to minimize its \emph{generalization error}
(or \emph{true loss}), which is defined to be
$ L_{\cal D}(w) = \Exp_{z \sim {\cal D}}[\ell(w, z)]$.
The minimum generalization error achievable over $\cal H$ is denoted as
$L_{\cal D}^* = \min_{w \in {\cal H}}L_{\cal D}(w)$.
In learning, $\cal D$ is \emph{unknown} but we are given a training set
$S = \{z_1, \dots, z_n \}$ drawn i.i.d. from $\cal D$.
The \emph{empirical loss function} is defined to be
$L_S(w) = \frac{1}{n}\sum_{i=1}^n\ell(w,z_i)$.

\cbstart
We are ready to define learnability. Our definition follows Shalev-Shwartz et al.
\cite{SSSSS10} which defines learnability in the Generalized Learning Setting
considered by Haussler~\cite{H92}. This definition also directly generalizes
PAC learnability~\cite{Valiant84}.
\begin{definition}[Learnability]
  A problem is called \emph{agnostically learnable with rate}
  $\varepsilon(n,\delta): \Nat \times (0,1) \mapsto (0,1)$
  if there is a learning rule $A: Z^n \mapsto {\cal H}$ such that for \emph{any}
  distribution $\cal D$ over $Z$, given $n \in \Nat$ and $\delta \in (0,1)$,
  with probability $1-\delta$ over $S \sim {\cal D}^{n}$,
  $L_{\cal D}(A(S)) \le L_{\cal D}^* + \varepsilon(n,\delta)$.
  Moreover, we say that the problem is \emph{agnostically learnable}
  if for any $\delta \in (0,1)$,
  $\varepsilon(n,\delta)$ vanishes to $0$ as $n$ tends to infinity.
\end{definition}
We stress that in this definition, the generalization error holds
{\em universally} for {\em every} distribution $\cal D$ on the data.\cbend
Intuitively, this definition says that given a confidence parameter
$\delta$ and a sample size $n$, the learned hypothesis $A(S)$ is
$\varepsilon(n,\delta)$ close to the best achievable. Note that
$\varepsilon(n,\delta)$ measures the rate we converge to the optimal.
Throughout this paper, we will use the \emph{generalization error} as the
\emph{utility measure} of a hypothesis $w$. Naturally, the utility of a
hypothesis is high if its generalization error is small.

In the above, the learning rule $A$ is \emph{deterministic} in the sense that
it maps a training set deterministically to a hypothesis in $\cal H$.
We will also talk about \emph{randomized} learning rules, which maps a training
set to a distribution over $\cal H$. More formally, a randomized learning rule
$\tildeA$ takes the form $Z^n \mapsto {\cal D}({\cal H})$, where
${\cal D}({\cal H})$ is the set of probability distributions over ${\cal H}$.
The empirical risk of $\tildeA$ on a training item $z$ is defined as
$\ell(\tildeA(S), z) = \Exp_{w \sim \tildeA(S)}[\ell(w, z)]$.
The empirical risk of $\tildeA$ on a training set $S$ is defined as
$L_S(\tildeA(S)) = \Exp_{w \sim \tildeA(S)}[L_S(w)]$.
Finally, we define the generalization error of $\widetilde{A}$ as
$L_{\cal D}(\tildeA(S)) = \Exp_{w \sim \tildeA(S)}[L_{\cal D}(w)]$.
In short, we take expectation over the randomness of $\tildeA$.

\vskip 5pt\noindent\textbf{Stability Theory}.
Stability theory is a sub-theory in machine learning
(see, for example, Chapter 13 of~\cite{UML} for a gentle survey of this area).
As we will see, more stable a learning rule is, better DP-utility tradeoff we
can achieve in a private learning.  To discuss stability, we need to define
``change of input data set.'' We will use the following definition.
\begin{definition}[Replace-One Operation]
  For a training set $S$, $i \in [n]$ and $z' \in Z$, we define $S^{(i,z')}$
  to be the training   set obtained by replacing $z_i$ by $z'$. In other words,
  \begin{align*}
    S &= \{z_1, \dots, z_{i-1}, z_i, z_{i+1}, \dots, z_n\},\\
    S^{(i,z')} &= \{z_1, \dots, z_{i-1}, z', z_{i+1}, \dots, z_n\}.
  \end{align*}
  Moreover, we write $S^{(i)}$ instead of $S^{(i,z')}$
  if $z'$ is clear from the context.
\end{definition}
Informally, a learning rule $A$ is stable if $A(S)$ and $A(S^{(i)})$ are
``close'' to each other. There are many possible ways to formulate what does
it mean by ``close.'' We will discuss the following definition, which is the
\emph{strongest} stability notion defined by Shalev-Shwartz et al.
\cite{SSSSS10}.
\begin{definition}[Strongly-Uniform-RO Stability~\cite{SSSSS10}]
  \label{def:strongly-uniform-ro-stability}
  A (possibly randomized) learning rule $A$ is \emph{strongly-uniform-RO stable}
  with rate $\varepsilon_{\rm stable}(n)$, if for all training sets $S$ of size
  $n$, for all $i \in [n]$, and all $z', \bar{z} \in Z$, it holds that
  $|\ell(A(S^{(i)}), \bar{z}) - \ell(A(S), \bar{z})|
   \le \varepsilon_{\rm stable}(n)$.
\end{definition}
Intuitively, this definition captures the following property of a good learning
algorithm $A$: if one changes \emph{any one training item} in the training set
$S$ to get $S'$, the two hypotheses computed by $A$ from $S$ and $S'$, namely
$A(S)$ and $A(S')$, will be ``close'' to each other (here ``close'' means that
for any instance $\bar{z}$ sampled from $\cal D$, the loss of $A(S)$ and
$A(S')$ on $\bar{z}$ are close).

A fundamental result on learnability and stability, proved recently by
Shalev-Shwartz et al.~\cite{SSSSS10}, states the following,
\begin{theorem}[\cite{SSSSS10}, informal]
  Consider any learning problem in the generalized learning setting
  as proposed by Vapnik~\cite{Vapnik98}.
  If the problem is learnable, then it can be learned by a (randomized) rule
  that is strongly-uniform-RO stable.
\end{theorem}
Qualitatively, this theorem says that
\emph{learnability and stability are equivalent}.
That is, every problem that is learnable can be learned \emph{stably}
(under strongly-uniform-RO stability).
In the context of differential privacy, this indicates that for any
learnable problem one might hope to achieve a good DP-utility tradeoff.
However, \emph{quantitatively} the situation is much more delicate:
as we will see later, strongly-uniform-RO stability is somewhat too weak to
lead to a differential privacy guarantee without additional assumptions.

\subsection{Convex Optimization}
\label{preliminaries:convex-optimization}
We will need the following basic concepts from convex optimization.
Let $\cal H$ be a closed convex set equipped with a norm $\| \cdot \|$.
A set $\cal H$ is $R$-bounded if $\|x\| \le R$ for any $x \in {\cal H}$.
A function $f: {\cal H} \mapsto \Real$ is \emph{convex}
if for every $u, v \in {\cal H}$, and $\alpha \in (0,1)$ we have
\begin{align*}
  f(\alpha u + (1-\alpha)v) \le \alpha f(u) + (1-\alpha)f(v).
\end{align*}
Moreover, $f$ is \emph{$\lambda$-strongly convex} if instead
\begin{align*}
  f(\alpha u + (1-\alpha)v) \le\ &\alpha f(u) + (1-\alpha)f(v) \\
  &- \frac{\lambda}{2}\alpha(1-\alpha)\|u-v\|^2.
\end{align*}
A function $f: {\cal H} \mapsto \Real$ is $\rho$-Lipschitz if for any
$u,v \in {\cal H}$, $|f(u) - f(v)| \le \rho\|u - v\|$.
Let $\cal H$ be a closed convex set in $\Real^d$.
A \emph{differentiable} function $f: {\cal H} \mapsto \Real$ is $\beta$-smooth
if its gradient $\nabla f$ is $\beta$-Lipschitz (note that $\nabla f$ is a
$d$-dimensional vector valued function).
That is, $\| \nabla f(u) - \nabla f(v) \| \le \beta \|u - v\|$.
We will use the following property of strongly convex functions.
It says that any convex function can be turned into a strongly convex one by
adding to it a strongly convex function.
\begin{lemma}[\cite{BV04}]
  \label{lemma:convex2strongly-convex}
  If $f$ is $\lambda$-strongly convex and $g$ is convex then $(f+g)$ is
  $\lambda$-strongly convex.
\end{lemma}

\subsection{Differential Privacy}
\label{preliminaries:dp}
Given two databases $D, D'$ of size $n$,
we say that they are \emph{neighboring} if they differ in at most one tuple.
We define bounded $\varepsilon$-differential privacy.
\begin{definition}[Bounded Differential Privacy]\ \\
  \label{def:bounded-dp}
  A mechanism ${\cal M}$ is called
  \emph{bounded $\varepsilon$-differentially private},
  if for any neighboring $D, D'$, and any event
  $E \subseteq {\rm Range}({\cal M})$,
  $\Pr[{\cal M}(D) \in E] \le e^{\varepsilon}\Pr[{\cal M}(D') \in E]$.
\end{definition}
This is called ``bounded'' differential privacy because the guarantee is over
databases of size $n$. Note that all previous work in differentially private
learning is for bounded differential privacy~\cite{BST14,CMS11,ZZXYW12}.
Dwork, McSherry, Nissim and Smith~\cite{DMNS06} provide an output perturbation
method for ensuring differential privacy. This method is based on estimating
the ``sensitivity'' of a query, defined as follows.
\begin{definition}
  \label{def:l2-sensitivity}
  Let $q$ be a query that maps a database to a vector in $\Real^d$.
  The $\ell_2$-sensitivity of $q$ is defined to be
  $$\Delta_2(q) = \max_{D \sim D'}\| q(D) - q(D') \|_2.$$
\end{definition}
Note that we are using the $2$-norm sensitivity instead of the usual $1$-norm.
We have the following theorem.
\begin{theorem}[\cite{DMNS06}]
  \label{fact:laplacian-dp}
  Let $q$ be a query that maps a database to a vector in $\Real^d$.
  Then publishing $q(D) + \kappa$ where $\kappa$ is sampled from the
  distribution with density
  \begin{align*}
    p(\kappa) \propto
    \exp\left(-\frac{\varepsilon}{\Delta_2(q)}\|\kappa\|_2\right)
  \end{align*}
  ensures $\varepsilon$-differential privacy.
\end{theorem}
Importantly, the $\ell_2$-norm of the noise vector, $\|\kappa\|_2$,
is distributed according to a Gamma distribution
$\Gamma\left(d, \frac{\Delta_2(q)}{\varepsilon}\right)$.
We have the following fact about Gamma distributions:
\begin{theorem}[\cite{CMS11}]
  \label{cor:norm-noise-vector}
  For the noise vector $\kappa$, we have that with probability at least
  $1-\gamma$, $\|\kappa\|_2 \le \frac{d\ln(d/\gamma)\Delta_2(q)}{\varepsilon}.$
\end{theorem}

\subsection{Our Setting}
\label{preliminaries:setting}
In this paper, we will work in the same setting as in~\cite{BST14},
where we assume the instance-wise loss function $\ell$ to be convex and
Lipschitz. As we will see later, both conditions are crucial for achieving
a good tradeoff between privacy and utility. We stress that \emph{both}
the convexity and Lipschitzness conditions here are
\emph{only with respect to $w$}. In other words, we only require
$\ell(w,z)$ to be convex and Lipschitz \emph{in $w$} (for any fixed
$z \sim {\cal D}$ at a time). Note also that, differing from Chaudhuri et al.
\cite{CMS11}, we do \emph{not} need $\ell$ to be differentiable in $w$.
Finally, we will focus on $\varepsilon$-differential privacy. Our results
readily extend to $(\varepsilon,\delta)$-differential privacy by using a
different noise distribution (e.g. Gaussian distribution) in output perturbation.


\section{Differential Privacy and Stability Theory}
\label{sec:dp-and-stability}
In this section we present our results on the connection
between differential privacy and stability theory.
Our technical results can be summarized as follows:
\vskip 5pt

\noindent\textbf{DP implies Strongly-Uniform-RO Stability}.
In Section~\ref{sec:dp-is-stability},
we show that, in the setting of learning,
differential privacy is a strong stability notion that
implies strongly-uniform-RO stability.
Strongly-Uniform-RO stability is the strongest stability notion
proposed by Shalev-Shwartz et al.~\cite{SSSSS10} from learning theory.
\vskip 5pt

\noindent\textbf{Norm Stability implies DP}.
In Section~\ref{sec:norm-stability},
we give a stability notion, $\ell_2$-RO stability,
that leads to differential privacy by injecting a small amount of noise.
$\ell_2$-RO stability is used implicitly in~\cite{SSSSS09} to prove
learnability of convex-Lipschitz-bounded problems under the condition
that the instance loss function is smooth.
Our analysis removes this requirement.
\vskip 5pt

\noindent\textbf{Simpler Mechanism with the Same Generalization Error}.
In recent work, Bassily et al.~\cite{BST14} give tight bounds
(for both training and generalization errors)
for differentially privately learning convex-Lipschitz-bounded problems.
Their mechanisms require the exponential mechanism and a sophisticated
sampling subprocedure. We show in Section~\ref{sec:dp-generalization} that
the elementary output perturbation mechanism presented in~\cite{CMS11}
can give the same tradeoff between differential privacy and
generalization error (however with weaker training error)
for \emph{every} convex-Lipschitz-bounded learning problem.
Our proof relaxes the technical requirements of~\cite{CMS11} (smoothness),
and is significantly simpler than both~\cite{BST14,CMS11}.

\subsection{Differential Privacy is a Stability Notion}
\label{sec:dp-is-stability}
If one writes out the definition of bounded differential privacy
(Definition~\ref{def:bounded-dp}) in the language of learning,
it becomes: a learning rule $A$ is $\varepsilon$-differentially private if,
for all training set $S$ of size $n$, for all $i \in [n]$, and all $z' \in Z$,
and \emph{any event $E$}, it holds that
$\Pr[A(S) \in E] \le e^{\varepsilon}\Pr[A(S^{(i)}) \in E]$,
where the probability is taken over the randomness of $A$.
When we contrast this definition with strongly-uniform-RO stability
(Definition~\ref{def:strongly-uniform-ro-stability}),
the only difference is that the latter considers a particular type of event,
namely for $\bar{z} \in Z$, the magnitude of $\ell(\tildeA(S), \bar{z})$.
At this point, it is somewhat clear that differential privacy is essentially
(yet another) stability notion.
Nevertheless, it is so strong that it implies strongly-uniform-RO stability,
as shown in the following:
\begin{proposition} 
  \label{prop:DP-implies-strongly-uniform-RO-stability}
  Suppose that $|\ell(\cdot,\cdot)| \le B$. Let $\varepsilon > 0$ and $A$ be a randomized learning rule.
  If $A$ is $\varepsilon$-differentially private,
  then it is strongly-uniform-RO stable with rate
  $\varepsilon_{\rm stable} \le B(e^\varepsilon-1)$.
  Specifically, for $\varepsilon \in (0,1)$, this is approximately $B\varepsilon$.
\end{proposition}
Two remarks are in order. First, this implication holds
without assuming anything on the loss function $\ell$ except for boundedness.
Second, one may note that the \emph{converse} of this proposition,
however, is not true in general.
For example, consider the case where $\ell$ is a constant function and
$A(S) = h_1 \neq h_2 = A(S^{(i)})$.
Then $A$ is strongly-uniform-RO stable (with rate $0$!)
yet it is clearly \emph{not} differentially private.
Moreover, this example indicates that,
even if strongly-uniform-RO stability has been achieved,
one cannot hope for differential privacy
by adding a ``small amount'' of noise to the output of $A$.
This is because $h_1$ and $h_2$ can be arbitrarily far away from each other
so the sensitivity of $A$ cannot be bounded.
This motivates us to define another stability notion for the purpose of
differential privacy.


\subsection{Norm Stability and Noise for DP}
\label{sec:norm-stability}
In this section we present a different stability notion that does lead to
differential privacy by injecting a small amount of noise.
We then use some machinery from stability theory to
quantify the amount of noise needed.
Following our discussion above, a natural idea now is that
$A(S)$ and $A(S^{(i)})$ shall be close \emph{by themselves},
rather than being close \emph{under the evaluation of some functions}.
Because the output of $A$ lies in $\cal H$, which is a normed space\footnote{
  For simplicity, $\| \cdot \|$ refers to $\ell_2$-norm in the rest of the paper.
  Our results are applicable to other settings as long as perturbation is properly
  defined on the normed space.}
as long as perturbation on ),
a ``universal'' notion for closeness is that
$A(S)$ and $A(S^{(i)})$ are close in norm.
This leads to the following definition.
\begin{definition}[$\ell_2$-RO Stability]
  A learning rule $A$ is \emph{$\ell_2$-RO stable} with rate $\varepsilon(n)$,
  if for \emph{any} $S \sim {\cal D}^n$, $z' \sim {\cal D}$ and $i \in [n]$,
  $\| A(S^{(i)}) - A(S) \|_2 \le \varepsilon(n)$.
\end{definition}
Astute readers may realize that this is nothing more than
a \emph{rephrasing of the $\ell_2$-sensitivity} of a query
(Definition~\ref{def:l2-sensitivity}).
Thus, in the spirit of the output perturbation method
mentioned in Section~\ref{preliminaries:dp},
if one can bound $\ell_2$-RO stability,
then we only need to inject a small amount of noise for differential privacy.

If $A$ is $\ell_2$-RO stable, then adding a small amount of noise
to its output ensures differential privacy,
and thus strongly-uniform-RO stability.
However, without the Lipschitz condition,
the resulting hypothesis might be useless.
This is because a small distance (in $\ell_2$-norm) to $A(S)$
could give significant change in loss.
This presents a barrier for proving learnability for the private mechanism.
Thus in the following, we will restrict ourselves back to
the setting discussed in Section~\ref{preliminaries:setting}
where we assume that $\ell$ is convex and Lipschitz (in $w$).

We now move on to quantifying the amount of noise needed
for differential privacy. Specifically, we will show that for strongly-convex
learning tasks, the ``scale of the noise'' we need is roughly only
$O_{d,\varepsilon}(1/n)$ where $n$ is the training set size
(the big-$O$ notation hides a constant
that depends on the number of features $d$, and the DP parameter $\varepsilon$.).
This means that as training set size increases,
the noise we need vanishes to zero for a fixed model and $\varepsilon$-DP.
By contrast, for the functional mechanism,
the ``scale of the noise'' is $O_{d, \varepsilon}(1)$.
The following two lemmas are due to Shalev-Shwartz et al.~\cite{SSSSS09}.
We include their proofs in the appendix for completeness.
\begin{lemma}[Exchanging Lemma]
  \label{lemma:exchanging}
  Let $A$ be a learning rule such that
  $A(S) = \arg\min_w\vartheta_S(w)$,
  where $\vartheta_S(w) = L_S(w) + \varrho(w)$
  and $\varrho(w)$ is a regularizer.
  For any $S \sim {\cal D}^n$, $i \in [n]$ and $z' \sim {\cal D}$,
  \begin{align*}
    \vartheta_S(u) - \vartheta_S(v)
    \le \frac{\ell(v, z') - \ell(u, z')}{n}
    + \frac{\ell(u, z_i) - \ell(v, z_i)}{n},
  \end{align*}
  where $u = A(S^{(i)})$ and $v = A(S)$.
\end{lemma}
Intuitively, this lemma concerns about the behavior of a learning rule $A$
on neighboring training sets. $A$ is a regularized learning rule:
Its objective function is in the form of empirical risk $L_S(w)$
plus regularization error $\varrho(w)$
($\varrho(\cdot)$ is called a regularizer).
More specifically, this lemma upper bounds
\emph{the difference between the objective values of $u$ and $v$},
that is $\vartheta_S(u)$ and $\vartheta_S(v)$,
in terms of instance losses on the specific two instances that get exchanged.

Recall that our goal is to upper bound $\|u - v\|$.
The following lemma accomplishes this task by
upper bounding the norm of the difference
by the difference of the objective values of $u$ and $v$.
\begin{lemma}
  \label{lemma:bound-norm-by-bounding-objective}
  Let $A$ be a rule where $A(S) = \arg\min_w\vartheta_S(w)$
  and $\vartheta_S(w)$ is $\lambda$-strongly convex in $w$.
  Then for any $S \sim {\cal D}^n$, $i \in [n]$ and $z' \sim {\cal D}$,
  $\frac{\lambda}{2} \| u - v\|^2 \le \vartheta_S(u) - \vartheta_S(v),$
  where $u = A(S^{(i)})$, $v = A(S)$.
\end{lemma}
To see this Lemma, we note that $v$ is a \emph{minimizer} of $\vartheta_S$
by the definition of the learning rule. Therefore by the definition of
strong convexity, we have that for any $\alpha > 0$,
\begin{align*}
  \vartheta_S(v) &\le \vartheta_S((1-\alpha)v + \alpha u) \\
  &\le \alpha \vartheta(v) + (1-\alpha)\vartheta(u)
  - \frac{\lambda}{2}\alpha(1-\alpha)\|u-v\|^2
\end{align*}
The lemma is then proved by rearranging and tending $\alpha$ to $1$.
Intuitively, this lemma says that as long as the objective function
$\vartheta_S(w)$ of $A$ is good (that is, strongly convex),
then one can upper bound the difference between $u$ and $v$ \emph{in norm}
by the difference between the {\em objective values} of $u$ and $v$.
Combining these two lemmas, we can prove the following main theorem
in this section.
\begin{theorem}
  \label{thm:l2-stability-Lip}
  Let $A$ be a learning rule with a $\lambda$-strongly convex
  objective loss function
  $\vartheta_S(w) = L_S(w) + \varrho(w)$
  where $\varrho(w)$ is a regularizer.
  Assume further that for any $z \in Z$,
  $\ell(\cdot, z)$ is $\rho$-Lipschitz.
  Then $A$ is $\frac{4\rho}{\lambda n}$ $\ell_2$-RO stable.
\end{theorem}
This theorem says that if both the \emph{instance loss function}
and the \emph{objective function} are well behaved
($\ell$ is $\rho$-Lipschitz and $\vartheta$ is strongly convex),
then the learning rule $A$ is roughly $(1/n)$-norm stable (in other words,
the sensitivity is $1/n$, which vanishes to $0$ as training set size $n$ grows).
\cbstart
In the case when $\ell$ is already $\lambda$-strongly convex,
we do not need a regularizer so one can set $\varrho(w) = 0$,
hence a natural algorithm to ensure differential privacy in this case
is to directly perturb the empirical risk minimizer
with noise calibrated to its norm stability.
Our discussion so far thus leads to
Algorithm~\ref{alg:output-perturbation-strongly-convex}

\begingroup
\removelatexerror
\begin{algorithm}[H]
  \removelatexerror
  \KwIn{Privacy budget: $\varepsilon_p > 0$.
    Training set $S = \{(x_i,y_i)\}_{i=1}^n$.}
  \KwOut{A hypothesis $w$.}
  \BlankLine
  Solve the empirical risk minimization $\Bar{w} = \arg\min_w L_S(w).$\\
  Draw a noise vector $\kappa \in \Real^d$ according to a distribution
  with density function
  $p(\kappa) \propto
  \exp\left(-\frac{\lambda n \varepsilon_p\|\kappa\|_2}{4\rho}\right).$\\
  Output $\Bar{w} + \kappa$.
  \caption{Output Perturbation for strongly-convex loss function:
    $\ell(w, z)$ is $\lambda$-strongly convex and $\rho$-Lipschitz in $w$,
    for every $z \in Z$.}
  \label{alg:output-perturbation-strongly-convex}
\end{algorithm}
\endgroup

\begin{theorem}
  \label{thm:dp-alg-strongly-convex}
  Let $(Z, {\cal H}, \ell)$ be a learning problem where $\ell$
  is $\lambda$-strongly convex and $\rho$-Lipschitz.
  Then Algorithm~\ref{alg:output-perturbation-strongly-convex}
  is $\varepsilon$-differentially private.
\end{theorem}
To see this, we note that $L_S(w)$ is $\lambda$-strongly convex
because $\ell$ is, so we can set the objective function $\vartheta(w) = L_S(w)$.
$\vartheta(w)$ is $\lambda$-strongly convex and $\rho$-Lipschitz,
so Theorem~\ref{thm:l2-stability-Lip} bounds its $\ell_2$-norm stability.
The proof is then completed by plugging the stability bound into
Theorem~\ref{fact:laplacian-dp}.

If the loss function is only convex (instead of being strongly convex),
then the idea is to use a strongly convex regularizer to
make the objective function strongly convex.
Specifically, we use the Tikhonov regularizer $\varrho(w) = \lambda\|w\|^2/2$
(indeed, any strongly convex regularizer applies).
This gives Algorithm~\ref{alg:output-perturbation-convex}.

\begingroup
\removelatexerror
\begin{algorithm}[H]
  \KwIn{Privacy budget: $\varepsilon_p > 0$.
    Regularization parameter: $\lambda > 0$.
    Boundedness parameter: $R > 0$.
    Training data: $S = \{(x_i,y_i)\}_{i=1}^n$.}
  \KwOut{A hypothesis $w$.}
  \BlankLine
  Solve the regularized empirical risk minimization problem
  $\Bar{w}
  = \arg\min_{\|w\| \le R} \left(L_S(w) + \frac{\lambda}{2}\|w\|^2\right).$\\
  Draw a noise vector $\kappa \in \Real^d$ according to a distribution where
  $p(\kappa) \propto
  \exp\left(
    -\frac{\lambda n \varepsilon_p\|\kappa\|_2}{4(\rho+\lambda R)}\right).$\\
  Output $\Bar{w} + \kappa$.
  \caption{Output Perturbation for general-convex loss function:
    $\ell(w, z)$ is convex and $\rho$-Lipschitz in $w$, for every $z \in Z$.}
  \label{alg:output-perturbation-convex}
\end{algorithm}
\endgroup

\begin{theorem}
  \label{thm:convex-dp-alg}
  Let $(Z, {\cal H}, \ell)$ be a learning problem where $\ell$
  is convex and $\rho$-Lipschitz.
  Suppose further that the hypothesis space is $\cal H$ is $R$-bounded.
  Then Algorithm~\ref{alg:output-perturbation-convex}
  is $\varepsilon$-differentially private.
\end{theorem}
The easiest way to see this theorem is to define a new loss function
$\Bar{\ell}(w, z) = \ell(w, z) + (\lambda\|w\|^2)/2$.
Then $\Bar{\ell}$ is $\lambda$-strongly convex
and $(\rho+\lambda R)$-Lipschitz over $\cal H$,
and the theorem directly follows from Theorem~\ref{thm:dp-alg-strongly-convex}.
We remark that Algorithm~\ref{alg:output-perturbation-convex} and
Theorem~\ref{thm:convex-dp-alg} can be strengthened based on
Theorem~\ref{thm:l2-stability-Lip} and
specifically do not rely on the boundedness condition.
However, since our analysis of generalization error for this case
(in the next section) critically relies on the boundedness condition,
the algorithm and its privacy guarantee are stated in the current form.\cbend


\subsection{Generalization Error}
\label{sec:dp-generalization}
Until now, we have only talked about
ensuring differential privacy with a small amount of noise,
and have not said anything about
whether this small amount of noise will lead to
a model with ``good utility'',
which is usually measured by \emph{generalization error} in learning theory.
We accomplish this in this section.
At a high level, we will show that for strongly convex learning tasks,
output perturbation produces hypotheses that are roughly $(1/n)$-away from the optimal.
For general convex learning tasks, this degrades to roughly $1/\sqrt{n}$.

For notational convenience,
throughout this section we let $A$ denote a deterministic learning mechanism,
and $\widetilde{A}$ be its output perturbation counterpart.
We begin with a general lemma.
\begin{lemma}
  \label{lemma:Lipschitz-condition-and-generalization-error}
  Suppose that for any $z \sim {\cal D}$, 
  $\ell(w, z)$ is $\rho$-Lipschitz in $w$.
  If with probability at least $1-\gamma$ over $w \sim \widetilde{A}(S)$,
  $$\|w - A(S)\|_2 \le \kappa(n, \gamma),\footnote{We abuse the notation
    $\kappa$ to remind readers that
    this quantity is related to the noise vector.}$$
  then with probability at least $1-\gamma$ over $w \sim \widetilde{A}(S)$,
  \begin{align*}
    |L_{\cal D}(w) - L_{\cal D}(A(S))| \le \rho\kappa(n, \gamma).
  \end{align*}
\end{lemma}
This lemma translates the closeness between two hypotheses \emph{in norm}
to the closeness in \emph{generalization error}.
More specifically, note that the randomized learning rule
$\widetilde{A}$ induces a \emph{distribution} $\widetilde{A}(S)$
over the hypothesis space.
This lemma says as long as a hypothesis sampled from
$\widetilde{A}(S)$ is close to $A(S)$ (a single hypothesis) {\em in norm},
then these two hypotheses are close in their generalization error.
Note that this closeness is controlled by the Lipschitz constant of
the instance loss function $\ell$.

Combing this lemma with Theorem~\ref{thm:l2-stability-Lip}
from the last section, which bounds the norm stability,
we have the following general theorem upper bounding
the generalization error of our method.
\begin{theorem}
  \label{thm:generalization-error-machinery}
  Let $({\cal H}, Z, \ell)$ be a learning problem that is agnostically learnable
  by a deterministic learning algorithm $A$ with rate $\varepsilon(n,\delta)$.
  Suppose that $\ell(w, z)$ is $\rho$-Lipschitz in $w$ for any $z \in Z$.
  Let $\cal D$ be a distribution over $Z$.
  Finally, suppose that for any $S \sim {\cal D}^n$,
  with probability at least $1-\gamma$ over $w \sim \widetilde{A}(S)$,
  \begin{align*}
    \|w - A(S)\|_2 \le \kappa(n, \gamma).
  \end{align*}
  Then with probability at least $1-\delta-\gamma$ over $S \sim {\cal D}^n$ and
  $w \sim \widetilde{A}(S)$, we have
  $L_{\cal D}(w) - L_{\cal D}^*
    \le \varepsilon(n,\delta) + \rho\kappa(n, \gamma).$
\end{theorem}
In the rest of this section we give concrete bounds
for different types of instance loss function.
\vskip 10pt

\noindent\textbf{Strongly Convex Loss}.
We now bound generalization error of our method for the case
where the instance loss function $\ell$ is strongly convex.
We will use the following theorem from stability theory:
\begin{theorem}[Theorem 6,~\cite{SSSSS09}]
  \label{thm:stochastic-strongly-convex-optimization}
  Consider a learning problem such that
  $\ell(w, z)$ is $\lambda$-strongly convex and $\rho$-Lipschitz in $w$.
  Then for any distribution $\cal D$ over $Z$ and any $\delta > 0$,
  with probability at least $1-\delta$ over $S \sim {\cal D}^n$,
  \begin{align*}
    L_{\cal D}(\erm(S)) - L_{\cal D}^* \le \frac{4\rho^2}{\delta\lambda n}.
  \end{align*}
  where $\erm(S)$ is the empirical risk minimizer,
  i.e. $\erm(S) = \arg\min_{w \in {\cal H}} L_S(w)$.
\end{theorem}
We are ready to prove the following theorem on generalization error.
Our bound matches the bound obtained by Bassily et al. for the same setting
(Theorem F.2,~\cite{BST14}).
\begin{theorem}
  \label{thm:strongly-convex-dp-generalization}
  Consider a learning problem such that $\ell(w, z)$ is
  $\lambda$-strongly convex and $\rho$-Lipschitz in $w$.
  Let $\varepsilon_p > 0$ be a privacy parameter for differential privacy.
  Let $\widetilde{A}$ denote Algorithm~\ref{alg:output-perturbation-strongly-convex}.
  Then for any $\delta \in (0,1)$,
  with probability $1-\delta$ over $S \sim {\cal D}^n$ and $w \sim \widetilde{A}(S)$,
  \begin{align*}
    L_{\cal D}(w) - L_{\cal D}^*
    \le O\Big(\frac{\rho^2d\ln(d/\delta)}{\lambda n \delta\varepsilon_p}\Big).
  \end{align*}
\end{theorem}
Basically, this theorem says that
if the instance loss function is strongly convex,
then with high probability, the generalization error of a hypothesis sampled
using our method is only roughly $1/n$-away from the ``optimal''
($L_{\cal D}^*$).
\vskip 10pt

\noindent\textbf{General Convex Loss}.
We now consider the general case where $\ell$ is only convex.
We have the following theorem on the generalization error,
which matches the bound obtained in~\cite{BST14} (Theorem F.3),
\begin{theorem}
  \label{thm:general-convex-dp-generalization}
  Consider a convex, $\rho$-Lipschitz learning problem that is also $R$-bounded.
  Let $\varepsilon_p > 0$ be a privacy parameter for differential privacy.
  Let $\widetilde{A}$ denote Algorithm~\ref{alg:output-perturbation-convex}.
  Then for any $\delta \in (0,1)$, with probability $1-\delta$
  over $S \sim {\cal D}^n$ and $w \sim \widetilde{A}(S)$,
  \begin{align*}
    L_{\cal D}(w) - L_{\cal D}^*
    \le O\left( \frac{\rho R\sqrt{d\ln(d/\delta)}}{\sqrt{n\delta\varepsilon_p}}
         \right).
  \end{align*}
\end{theorem}
Note that for convex loss functions (instead of strongly convex ones),
we are only roughly $1/\sqrt{n}$-away from the optimal.

We notice that regularization plays a vital role in this theoretical guarantee.
Indeed, the key to prove Theorem~\ref{thm:general-convex-dp-generalization} is
the use of the Tikhonov regularizer ($\lambda\|w\|^2/2$) to gain {\em stability}.
By contrast, regularization is only used as a heuristic
in the analysis of the functional mechanism~\cite{ZZXYW12}.



\section{Applications}
\label{sec:applications}
\cbstart
In this section we give concrete applications of the theory developed so far.
Along the way, we compare our method with previous ones.
To begin with, Chaudhuri et al.~\cite{CMS11} have given
an output perturbation mechanism that is essentially the same as our
Algorithm~\ref{alg:output-perturbation-convex}.
Moreover, they give an {\em objective perturbation} algorithm.
The main difference is that the noise is added to the objective function,
instead of the output. Algorithm~\ref{alg:objective-perturbation-convex}
describes their objective perturbation algorithm.
We refer interested readers to their work for more details.

\begingroup
\removelatexerror
\begin{algorithm}[H]
  \KwIn{Privacy budget: $\varepsilon_p > 0$.
    Parameters $\lambda, c > 0$.
    Boundedness parameter: $R > 0$.
    Training data: $S = \{(x_i,y_i)\}_{i=1}^n$.}
  \KwOut{A hypothesis $w$.}
  \BlankLine
  Let $\varepsilon_p'
  = \varepsilon_p - \log(1 + \frac{2c}{n\lambda} + \frac{c^2}{n^2\lambda^2})$.\\
  If $\varepsilon_p' > 0$, then $\Delta = 0$,
  else $\Delta = \frac{c}{n(e^{\varepsilon p/4})-1} - \lambda$,
  and $\varepsilon_p'=\varepsilon_p/2$.\\
  Put $\beta = \varepsilon_p/2$ and draw noise vector $\nu$
  such that $p(\nu) \propto e^{-\beta\|\nu\|}.$\\
  Output
  $\arg\min_{\|w\| \le R}
  \Big\{
  L_S(w) + \frac{\nu^Tw}{n} + \frac{\lambda\|w\|^2}{2}
  \Big\}.$
  \caption{Objective Perturbation of Chaudhuri et al.}
  \label{alg:objective-perturbation-convex}
\end{algorithm}
\endgroup

Our work differs in analysis and applicability.
Specifically, to get their claimed generalization bounds
Chaudhuri et al. requires differentiability and the Lipschitz {\em derivative}
for their output perturbation, and requires twice differentiability
with bounded derivatives for their objective perturbation
(see Theorem 6 and Theorem 9 in~\cite{CMS11}).
Our analysis removes all these differentiability conditions,
and demonstrates that {\em output perturbation} works well
for {\em all} convex-Lipschitz-bounded learning problems.
This is by far the largest class of convex learning problems
that are known to be learnable (see~\cite{UML}).
Due to lack of space, in the following we only give two important examples,
Support Vector Machine and Logistic Regression.
As we will see, our analysis enables us to train an
SVM {\em without} approximating the loss function
using a smooth loss function that gives higher or equal loss point-wisely.
\vskip 5pt
\noindent\textbf{Support Vector Machine (SVM)}.
In SVM the instance-loss function is the so called {\em hinge loss function}.
Specifically, given hypothesis space ${\cal H} \subseteq \Real^d$,
feature space $X \subseteq \Real^d$, and output space $Y = \zo$,
then for $(x,y) \in X \times Y$, the hinge loss is defined as
\[ \ell^{\rm hinge}(w, x, y) = \max\big\{0, y(1-\langle w, x \rangle)\big\}.\]
In the common setting where $X$ is scaled to lie within the unit ball,
it is straightforward to verify that $\ell^{\rm hinge}$ is convex $1$-Lipschitz.
Therefore our Algorithm~\ref{alg:output-perturbation-convex}
and Theorem~\ref{thm:general-convex-dp-generalization} directly apply to give
differential privacy with a small generalization error. We note that, however,
hinge loss is {\em not} differentiable. Therefore Chaudhuri et al.'s analysis
for output perturbation does not directly apply.
Indeed, in order to use their method, they need to use a
smooth loss function to approximate the hinge loss
by giving higher or equal loss point-wisely.
\vskip 5pt
\noindent\textbf{Logistic Regression}. In the Logistic Regression
we have ${\cal H} \subseteq \Real^d$, $X \subseteq \Real^d$ and $Y = \{ \pm 1 \}$.
Then for $y \in Y$ and $x \in X$, the loss of $w$ on $(x,y)$ is defined to be
\[ \ell^{\rm log}(w, x, y) = \log\big( 1 + e^{-y\langle w, x \rangle} \big). \]
It is not hard to verify that $\ell^{\rm log}$ is convex and differentiable
with bounded derivatives. Therefore both our analysis and Chaudhuri et al.'s
analysis directly apply. In particular, Chaudhuri et al.~\cite{CMS11}
has done extensive experiments evaluating output perturbation
for Logistic Regression on standard datasets.

\subsection{On Setting Parameters, Generalization Error and Implementation}
\label{sec:theory-to-algorithms:parameter-tuning}
We note that, very recently, Bassily et al.~\cite{BST14} have given better
{\em training error} under the {\em same} technical conditions as ours
(the generalization error is the same).
Next we compare with theirs.

On one hand, we observe that~\cite{BST14, CMS11}, as well as our work,
are all based on regularized learning,
and the regularization parameter $\lambda$ is left unspecified.
Because the requirement of differential privacy, one cannot na\"{i}vely run
the regularized learning on multiple training sets and then pick the best parameters:
Our analysis of DP is with respect to the training data only;
however, if the parameters we pick depend on the data,
then they may carry sensitive information of individuals in the data,
which renders the subsequent DP analysis invalid.
A standard way to tackle this issue is to employ private parameter tuning,
which acts as a {\em wrapper procedure},
that invokes the regularized learning procedure as a black box,
and pick parameters differentially privately for the regularized learning.
In this paper we also consider a {\em data independent} approach based on
our analysis.

On the other hand, for fixed $\lambda$ the regularized learning procedure
of Bassily et al.~\cite{BST14} achieves better training error
(the generalization error remains the same).
However, their mechanisms are substantially more complicated.
More specifically, their mechanism requires a sophisticated sampling procedure
that runs in time $O(n^3)$, where $n$ is the training set size.
For reasonably large data sets this can be prohibitive.
An even more serious concern is some practical challenges
in implementing this sampling procedure,
which dates back to work of Applegate and Kannan~\cite{AK91}.
Indeed, we noticed that similar concern in implementing
related sampling procedures have also been raised by Hardt~\cite{Hardt14}.
Given these concerns,
we will not empirically compare with Bassily et al.'s mechanisms.
\cbend

\subsection{Linear Regression}
\label{sec:theory-to-algorithms:linear-regression}
In the rest of this paper, we specialize the general algorithms
to linear regression, which is a common regression task and is the model which
Fredrikson et al.~\cite{FLJLPR14} constructed from a medical data set.
We will evaluate the effectiveness of our private linear regression mechanisms
presented here in an empirical study later.

Technically, let $d$ be the number of features,
${\cal H} \subseteq \Real^d$, $X \subseteq \Real^d$ and $Y \subseteq \Real$.
Linear regression uses the so called ``\emph{squared loss}'' 
instance-loss function, $\ell(w, (x, y)) = (\langle w, x \rangle - y)^2$.
Note that $\ell$ is convex in $w$ for any fixed $(x,y)$.
Unfortunately, it is known in the learning theory that the learning task
with respect to $\ell$ (that is, to minimize $\Exp_{\cal D}[\ell(w, {\cal D})]$
over $w \in {\cal H}$) is not learnable over $\Real^d$ (see Example 12.8,
Shalev-Shwartz and Ben-David~\cite{UML}.).
To gain theoretical tractability, we thus restrict
the hypothesis space $\cal H$ to be $R$-bounded.
With this restriction, $\ell$ is then $(2R+2)$-Lipschitz over $\cal H$.
Thus we can specialize Algorithm~\ref{alg:output-perturbation-convex} to
get $\outalg$ described in Algorithm~\ref{alg:output-perturbation-linear-regression}.

\begin{algorithm}[!htp]
  \KwIn{Privacy budget: $\varepsilon_p > 0$.
    Regularization parameter: $\lambda > 0$.
    Boundedness parameter: $R > 0$.
    Training data: $S = \{(x_i,y_i)\}_{i=1}^n$,
    where  $\|x_i\| \le 1$ and $|y_i| \le 1$.}
  \KwOut{A hypothesis $w$.}
  \BlankLine
  Solve the regularized least mean square problem
  $$\Bar{w}
  = \arg\min_{\|w\| \le R} \left(L_S(w) + \frac{\lambda}{2}\|w\|^2\right).$$\\
  Draw a noise vector $\kappa \in \Real^d$ according to a distribution
  with the following density function $p(\kappa)$,
  $$p(\kappa) \propto
  \exp\left(-\frac{\lambda n \varepsilon\|\kappa\|_2}{12R+8}\right).$$\\
  Output $\Bar{w} + \kappa$.
  \caption{\outalg}
  \label{alg:output-perturbation-linear-regression}
\end{algorithm}

Note that $\outalg$ has two unspecified parameters $\lambda$ and $R$.
We use two approaches to tackle the issue of picking these parameter privately:
The first is to use a computation based on our analysis to pick parameters
in a data {\em independent} manner.
The second approach uses a private parameter tuning algorithm
(Section 6, \cite{CMS11}).
We are ready to describe all the private linear regression mechanisms
considered in this paper and explain on how to choose parameters.
In the following, similar to previous work~\cite{CMS11,ZZXYW12},
we assume that the training data is scaled so that they lie in the unit ball.
\vskip 5pt

\noindent\textbf{Oracle Output Perturbation}.
This is the \emph{idealized} version of the output perturbation mechanism,
where we exhaustively search for the best parameters using the training data
and then use them in as if they were ``constants''.
We consider this variant because it shows the best possible result one can
obtain using output perturbation.

We search $R$ exhaustively within the space $\{0.25, 0.5, 1, 2\}$.
To search for a good $\lambda$, we consider an arithmetic progression
starting at initial value $(d/n\varepsilon'_p)^{1/2}$
and ending at $(d/n\varepsilon^*_p)^{1/2}$,
where $\varepsilon'_p=100$ and $\varepsilon^*_p=0.1$.
The start and end values are determined by a computation based on our proof
of Theorem~\ref{thm:general-convex-dp-generalization}.
For our dataset on warfarin dosage, $n \approx 3000$ and $d = 14$.
This gives the approximate range $[0.007, 0.2]$.
We use a slightly larger range $[0.001, 0.5]$.
\vskip 5pt

\noindent\textbf{Data-Independent Output Perturbation}.
Based on the proof of Theorem~\ref{thm:general-convex-dp-generalization},
We determine $\lambda, R$ through a computation so that
they are independent of the training data.

We put $R=1$. This is because we scale the data to be within the unit ball,
and the coefficients of the linear regression model are all small
when trained over the scaled data.
For $\lambda$, we pick it using our theoretical analysis for output perturbation.
Specifically, Theorem~\ref{thm:general-convex-dp-generalization} indicates
that $\lambda$ is approximately $\sqrt{d/n\varepsilon_p}$.
\vskip 5pt

\cbstart
\noindent\textbf{Privately-Tuned Output/Objective Perturbation}.
Chaudhuri et al.~\cite{CMS11} use private parameter tuning to
pick data-dependent parameters for both of their
output/objective perturbation.
Fortunately, the private tuning algorithm only makes black-box use of
the output/objective perturbation.
Algorithm~\ref{alg:tuning-algorithm} describes the private-tuning framework.
Given $\varepsilon_p$, the tuning algorithm ensures $\varepsilon_p$-DP.
\cbend
\begin{algorithm}[!htp]
  \KwIn{Privacy budget: $\varepsilon_p > 0$.\\
    Training data: $S = \{(x_i,y_i)\}_{i=1}^n$,
    where $\|x_i\| \le 1$, $|y_i| \le 1$.\\
    Parameter space: ${\cal R} = \{R_1,\dots,R_l\}$,
    where $\max_{1 \le i \le l}R_i \le R$,
    $\Lambda = \{\lambda_1,\dots,\lambda_k\}$.
    Let ${\cal R} \times \Lambda = \{(R_i,\lambda_i)\}_{i=1}^m$}
  \KwOut{A hypothesis $w$.}
  \BlankLine
  Divide the training data into $m+1$ chunks, $S_1,\dots,S_m,S'$.
  $S_1,\dots,S_m$ are used for training, and $S'$ is used for validation.\\

  For each $i=1,2,\dots,m$,
  apply a privacy-preserving algorithm to train $w_i$
  (for example output or objective perturbation) with parameters
  $\varepsilon_p, \lambda_i, R_i, S_i$.
  Evaluate $w_i$ on $S'$ to get utility $u_{S'}(w_i) = -L_{S'}(w_i)$.\\

  Pick a $w_i$ in $\{w_1,\dots,w_m\}$
  using the exponential mechanism with privacy budget $\varepsilon_p$ and
  utility function $u_{S'}(w_i)$.
  Note that the sensitivity of $u$ is
  $$\Delta(u)
  = \max_{w \in {\cal H}}\max_{S,i,z'}|u_S(w) - u_{S^{(i)}}(w)| \le R^2.$$
  Thus this amounts to sampling $w_i$ with probability
  \begin{align*}
    p(w_i) \propto \exp\left(-\frac{L_{S'}(w_i)\varepsilon_p}{2R^2}\right).
  \end{align*}
  \caption{Parameter Tuning Algorithm:
    it accesses a privacy-preserving training algorithm as a black box.}
  \label{alg:tuning-algorithm}
\end{algorithm}

Let $\cal R$ be a set of $l$ choices of $R$ and $\Lambda$ be
a set of $k$ choices of $\lambda$.
Let $m = l \cdot k$, and suppose the pairs of parameters can be listed as
$\{(R_1,\lambda_1),\dots,(R_m,\lambda_m)\}$.
The more settings of parameters to try, larger the $m$ is,
and so we have smaller chunks for training, thus more entropy
in the probability in picking the final hypothesis.
Therefore, we want to have a small set of good parameters.
We put $\Lambda$ as a geometric progression of ratio $2$ that starts with
$0.002$ and ends with $0.256 = 0.002 \cdot 2^7$, so $k = 8$.
For $R$, note that while we may want larger radius for the purpose of learning,
the probability we assign to each $w_i$ decays quadratically in $R$.
We thus try two sets ${\cal R}_1 = \{.25, .5, 1\}$,
and ${\cal R}_2 = \{.5, 1, 2\}$, so $l=3$.
For ${\cal R}_1$, the probability sampling $w_i$ becomes
$\exp(-L_{S'}(w_i)\varepsilon_p/2)$.
For ${\cal R}_2$, this is $\exp(-L_{S'}(w_i)\varepsilon_p/8)$.
\vskip 5pt


\section{Empirical Study}
\label{sec:evaluation}

In this section we evaluate our mechanisms on the warfarin-dosing dataset used
by Fredrikson et al.~\cite{FLJLPR14}. We are interested in the following
questions:
\begin{enumerate}
\item How does the accuracy of models produced using our mechanisms compare to
the functional mechanism?
\item What is the impact of improved DP-utility tradeoff on MI attacks?
\end{enumerate}
In summary, we found that: \textit{1)} our mechanisms provide better accuracy
than the functional mechanism, for a given $\epsilon$ setting, and \textit{2)}
this improvement in model accuracy makes MI attacks more effective.
Specifically, we obtain accurate models with $\varepsilon$-DP even for
$\varepsilon = 0.1$, whereas the functional mechanism does not provide
comparable models until $\varepsilon \ge 5$.

Section~\ref{sec:evaluation:setup} describes our experimental methodology, and
Section~\ref{sec:evaluation:tradeoff-dp-utility} presents our results on the
relationship between $\varepsilon$ and model accuracy. In
Section~\ref{sec:evaluation:model-inversion}, we present empirical evidence
and a formal analysis of our observation that better differentially-private
mechanisms lead to more effective MI attacks.

\subsection{Methodology}
\label{sec:evaluation:setup}
We used the same methodology, data, and training/validation split as the
experiments discussed in Fredrikson et al.~\cite{FLJLPR14}. The data was
collected by the International Warfarin Pharmacogenetics Consortium (IWPC),
and contains information pertaining to the age, height, weight, race, partial
medical history, and two genomic SNPs: VKORC1 and CYP2C9. The outcome variable
corresponds to the stable therepeutic dose of warfarin, defined as the
steady-state dose that led to stable anticoagulation levels. At the time of
collection, this was the most expansive database of information relevant to
pharmacogenomic warfarin dosing. For more information about the data and how
it was pre-processed, we refer the reader to the original IWPC
paper~\cite{IWPC09} and the paper by Fredrikson et al.~\cite{FLJLPR14}.

Our experiments examine linear warfarin dosing models trained on this data,
either using differentially-private regression mechanisms or standard linear
regression, and seek to characterize the relationship between the following
quantities:
\vskip 5pt
\noindent\textbf{Privacy budget}. This is given by the parameter $\epsilon$, and
controls the ``strength'' of the guarantee conferred by differential privacy.
Smaller privacy budgets correspond to stronger guarantees.
\vskip 5pt

\noindent\textbf{Model accuracy/utility}. We measure this as the average
mean-squared error between the model's predicted warfarin dose, and the ground
truth given in the validation set. To obtain a more realistic measure of
utility, we use the approach described by Fredrikson et al.~\cite{FLJLPR14} to
simulate the physiological response of patients who are are given warfarin
dosages according to the model, and estimate the likelihood of adverse health
events (stroke, hemorrhage, and mortality).
\vskip 5pt

\noindent\textbf{Model invertability}. This is measured by the success rate of the
model inversion algorithm described by Fredrikson et al.~\cite{FLJLPR14}, when
predicting VKORC1 for each patient in the validation set. We focused on VKORC1
rather than CYP2C9, as Fredrikson et al. showed that the results for this SNP
more clearly illustrate the trade-off between privacy and utility.
\vskip 5pt

The rest of this section describes our results for each quantity.

\subsection{Model Accuracy}
Figure~\ref{fig:out-fun-mse} compares Functional Mechanism with all the
private output perturbation mechanisms,
which includes Tuned Output Perturbation (Out),
Data-Independent Output Perturbation (Data Independent),
and Oracle Output Perturbation (Oracle).
We also include the model accuracy of the non-private algorithm (Non-Private).
We observe that all the output perturbation give much better model accuracy
compared to the functional mechanism, especially for small $\varepsilon$.
Specifically, Tuned Output Perturbation obtains accurate models at
$\varepsilon = 0.3$.
Data-Independent and the Oracle Mechanisms give much the same model accuracy,
and provide accurate models even for $\varepsilon = 0.1$.
In particular, the accuracy is very close to the models
produced by the  non-private algorithm.
\label{sec:evaluation:tradeoff-dp-utility}
\begin{figure}[htp]
  \centering
  \includegraphics[width=0.6\columnwidth]{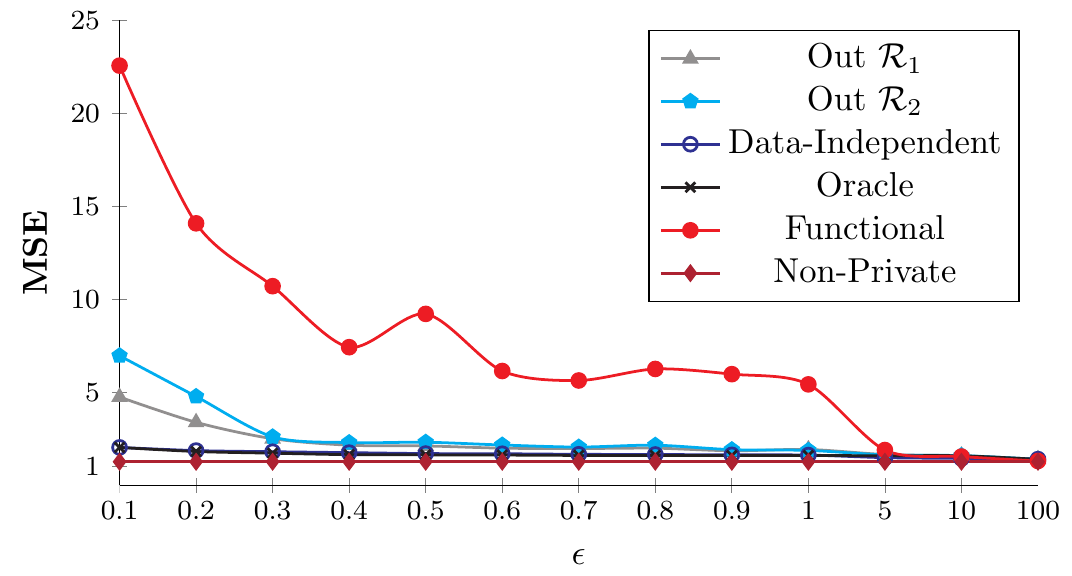}
  \vspace{-1\baselineskip}
  \caption{
    Model Accuracy: compare the functional mechanism with
    our output perturbation method, as well as the non-private mechanism.
    ${\cal R}_1$ stands for the Private-Tuned mechanism where we try
    the bounded hypothesis space with radius $R$ in $\{.25, .5, 1\}$.
    ${\cal R}_2$ stands for the same mechanism with radius in $\{.5, 1, 2\}$.
    Out indicates Privately Tuned Output Perturbation methods.
    Oracle (Data-Independent) stands for the
    Oracle (Data-Independent) Output Perturbation.
  }
  \label{fig:out-fun-mse}
\end{figure}

Figure~\ref{fig:di-oracle-mse} further compares the performance of
Data-Independent and Oracle mechanisms and demonstrate their closeness.
For the entire parameter range considered,
the maximum MSE gap we observed is only $0.1$.
\begin{figure}[htp]
  \centering
  \includegraphics[width=0.6\columnwidth]{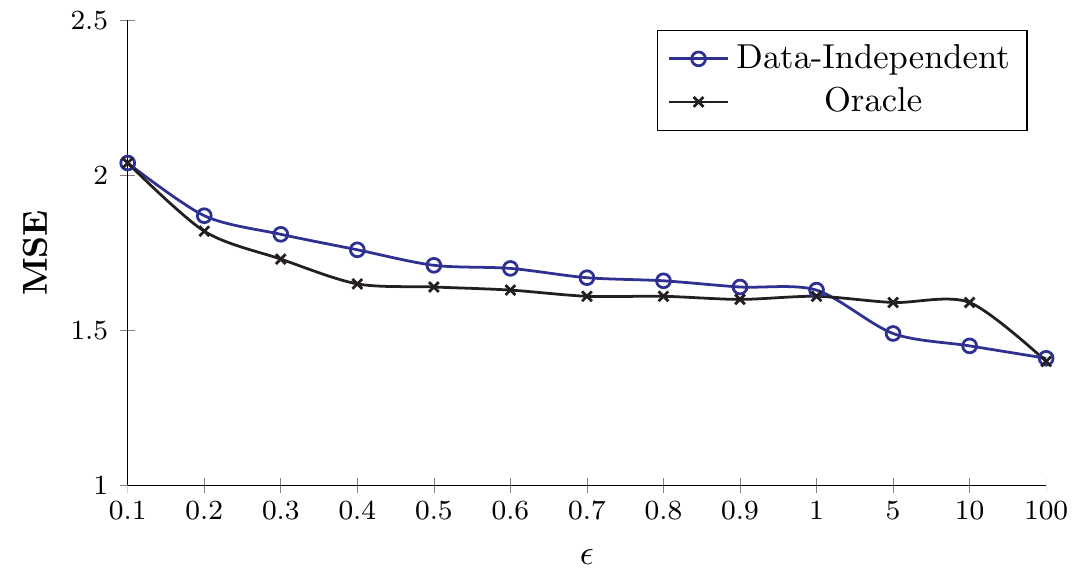}
  \vspace{-1\baselineskip}
  \caption{
    Model Accuracy: Further comparison between Data-Independent Output Perturbation
    and Oracle Output Perturbation. The maximum MSE gap we observed is only $0.1$.
  }
  \label{fig:di-oracle-mse}
\end{figure}
Figure~\ref{fig:risk} further compares the risk of mortality,
hemorrhage, and stroke using Functional Mechanism,
Tuned Output Perturbation and Data-Independent Output Perturbation.
unsurprisingly, Data-Independent Output Perturbation gives the best result,
and in particular much smaller risk than Functional Mechanism.
\begin{figure*}[htp]
  \centering
  \subfloat[Mortality]{
    \includegraphics[width=.32\columnwidth]{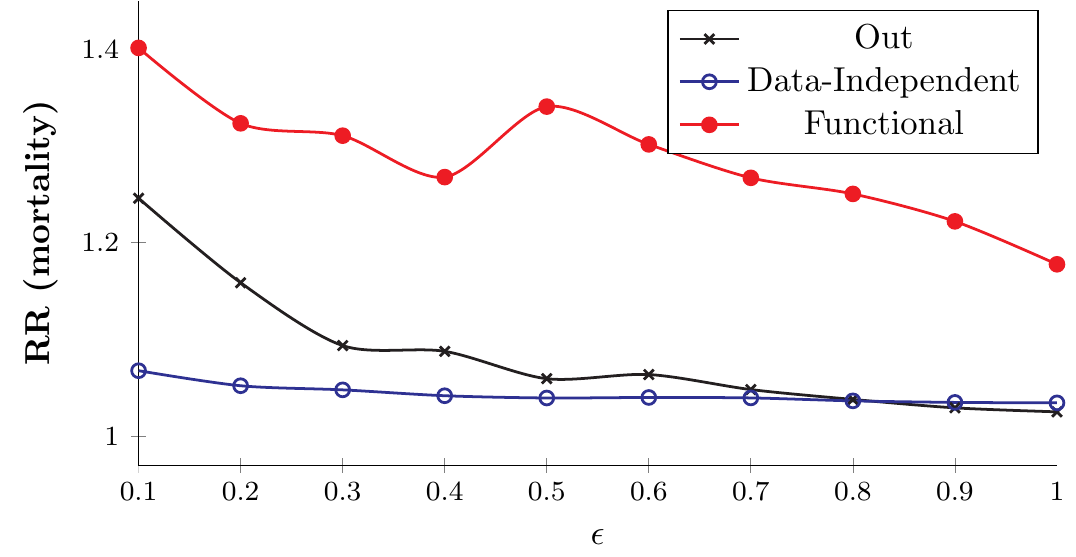}}
  \subfloat[Bleeding]{
    \includegraphics[width=.32\columnwidth]{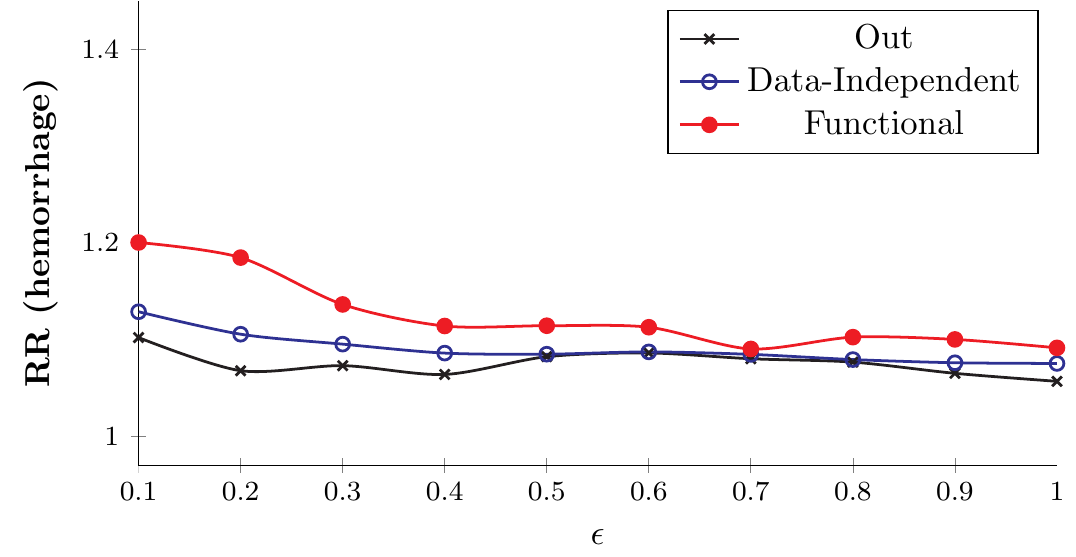}}
  \subfloat[Stroke]{
    \includegraphics[width=.32\columnwidth]{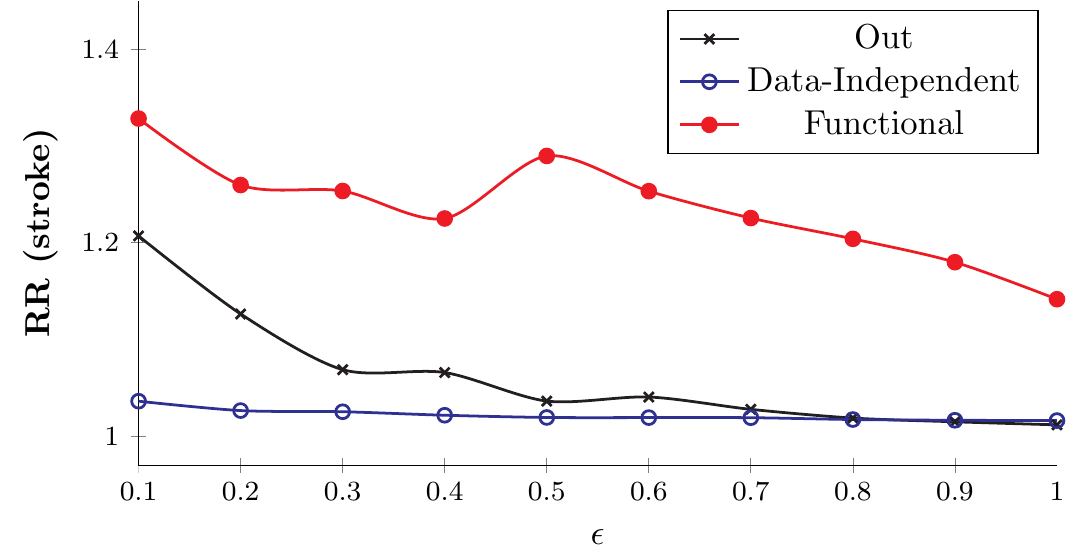}}
  \vspace{-1\baselineskip}
  \caption{Model Risk: We estimate the risk of mortality, hemorrhage,
    and stroke using the approach described by Fredrikson et al.}
  \label{fig:risk}
\end{figure*}
\vskip 5pt
\cbstart
\noindent\textbf{Comparison with Other Private Mechanisms}.
We also compare model accuracy of our method with two other private algorithms.
\vskip 5pt
\noindent\textbf{Projected Histogram}.
We notice that in the previous work of Fredrikson et al.~\cite{FLJLPR14},
they have implemented private projected histogram mechanism and compared it
with the functional mechanism on linear regression.
Specifically, as Figure 6 (Section 5.2) of their paper shows,
the projected-histogram algorithm indeed has similar model accuracy
compared to the functional mechanism,
which is much worse than our Data-Independent algorithm
(of which the accuracy is close to that of the non-private algorithm).
\vskip 5pt
\noindent\textbf{Objective Perturbation}.
We have so far mainly compared with the function mechanism,
which we believe is the most important task,
because the functional mechanism has become the recognized state of the art
for training regression models and has been adopted by many research teams,
including~\cite{APW15,WSW15,WYZ12}. On the other hand,
we notice that under very restricted conditions,
it is known that Chaudhuri et al.'s objective perturbation method~\cite{CMS11}
can provide very good model accuracy.
Interestingly, because we impose boundedness condition for linear regression,
the technical conditions of objective perturbation are satisfied.
Therefore we also compare it with our method.
Encouragingly and perhaps somewhat surprisingly,
while our method is much more widely applicable
(see discussion in Section~\ref{sec:applications}),
our experiments show that our general algorithm performs
as well as objective perturbation.
Specifically, Figure~\ref{fig:obj-di-oracle-mse} compares the model accuracy of
these three methods, and the maximal gap we observed is only $0.1$!
\begin{figure}[!htp]
  \centering
  \includegraphics[width=0.6\columnwidth]{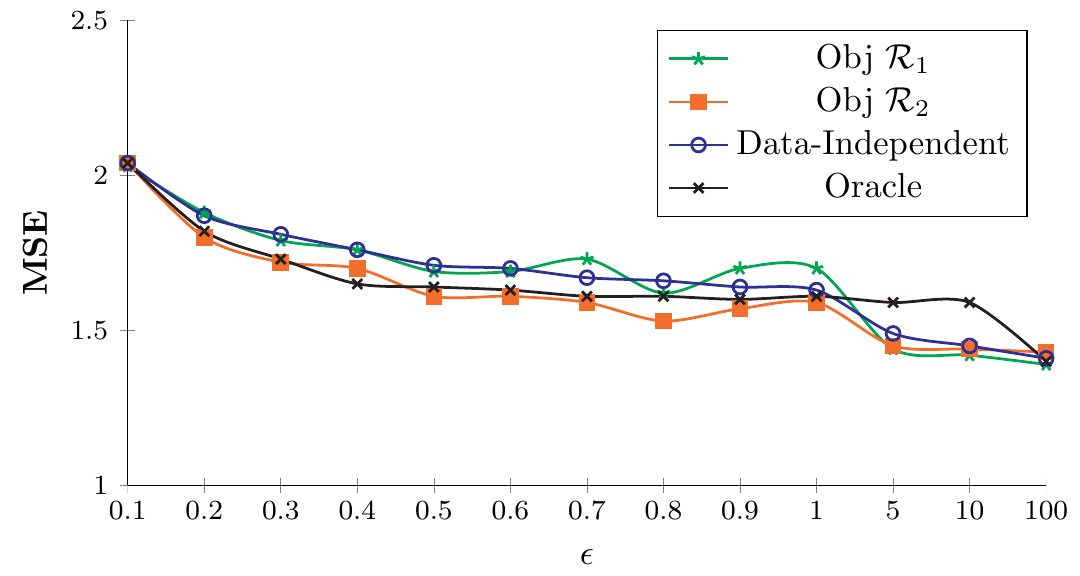}
  \vspace{-1\baselineskip}
  \caption{
    Model Accuracy: Objective Perturbation,
    Data-Independent Output Perturbation and Oracle Output Perturbation.
    ${\cal R}_1$ stands for the Private-Tuned mechanism where we try
    the bounded hypothesis space with radius $R$ in $\{.25, .5, 1\}$.
    ${\cal R}_2$ stands for the same mechanism with radius in $\{.5, 1, 2\}$.
  }
  \label{fig:obj-di-oracle-mse}
\end{figure}
\cbend

\subsection{Model Inversion}
\label{sec:evaluation:model-inversion}
In this section, we examine the impact of the increased model utility on
\emph{model inversion} (MI), a privacy attack first raised by Fredrikson et
al.~\cite{FLJLPR14}. Improving model utility for a given $\varepsilon$ is a
theme shared by nearly all previous work on differential privacy. This is a
sensible goal, because utility has no direct bearing on the privacy guarantee
provided by \emph{differential-privacy}---two models can differ significantly
on the level of utility the provide, while still conferring the same level of
differential privacy. However, we show that MI is orthogonal to differential
privacy in this sense, because the improved utility offered by our mechanisms
leads to more successful MI attacks.
\vskip 5pt
\noindent\textbf{Better DP mechanisms, more effective MI attacks.}
Figure~\ref{fig:out-fun-mi} compares MI accuracy of all the private mechanisms.
For all these mechanisms, we see that mechanisms with better DP-utility tradeoff
also has higher MI accuracy.
Specifically,  For Oracle Output Perturbation at $\varepsilon=0.2$,
we see a significant increase in MI accuracy:
$45\%$ for Oracle compared to $35\%$ for the functional mechanism.
Meanwhile, the utility of our mechanism is much better,
with mean squared error $1.82$ compared to the functional mechanism's $22.57$.
This phenomenon holds for larger $\varepsilon$,
although the magnitude of the differences gradually shrink.

\begin{figure}[htp]
  \centering
  \includegraphics[width=0.6\columnwidth]{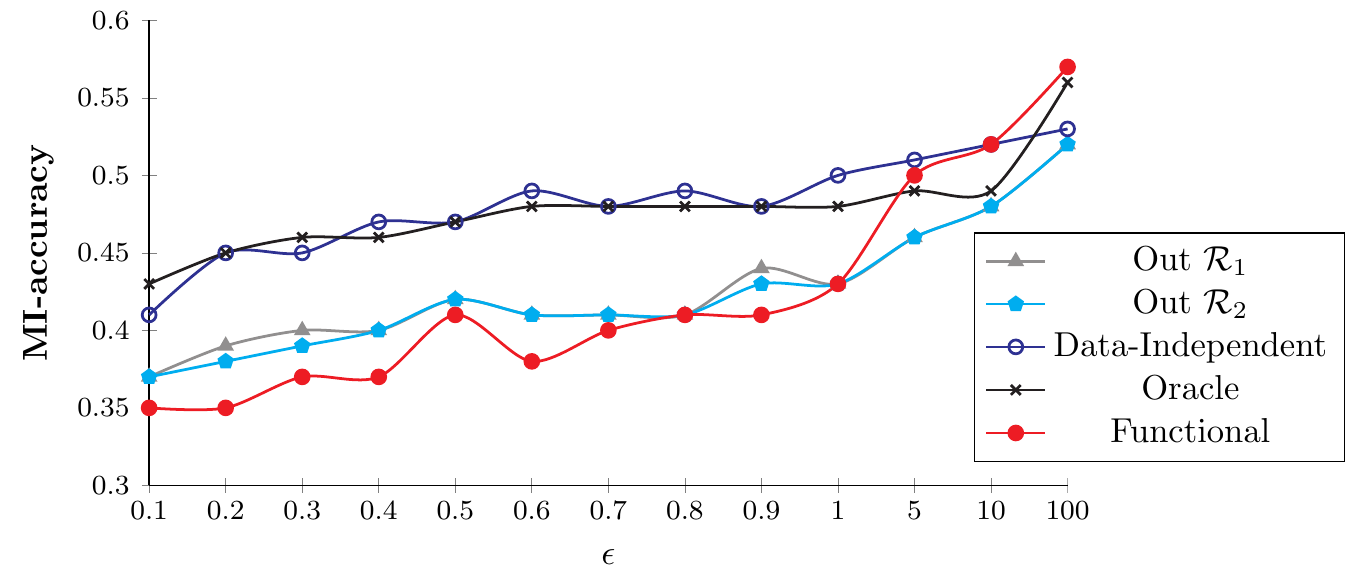}
  \vspace{-1\baselineskip}
  \caption{
    MI attack: Output Perturbation vs. Functional Mechanism.
    The x-axis is $\varepsilon$. The y-axis is the accuracy of MI attack.
    MI attacks become more effective for all three variants. 
    For Oracle Output Perturbation and Functional Mechanism at $\varepsilon=.2$,
    the MI attack accuracy for the oracle mechanism is $45\%$,
    while is only $35\%$ for the functional mechanism.
    Data-Independent and Oracle Output Perturbation are similar.
  }
  \label{fig:out-fun-mi}
\end{figure}

Figure~\ref{fig:di-oracle-mi} further demonstrates that,
similar to their model accuracy,
Data-Independent and Oracle Output Perturbation
have very similar behavior on model invertibility
(note that they both provide similar model accuracy
that is better than other methods).

\begin{figure}[htp]
  \centering
  \includegraphics[width=0.6\columnwidth]{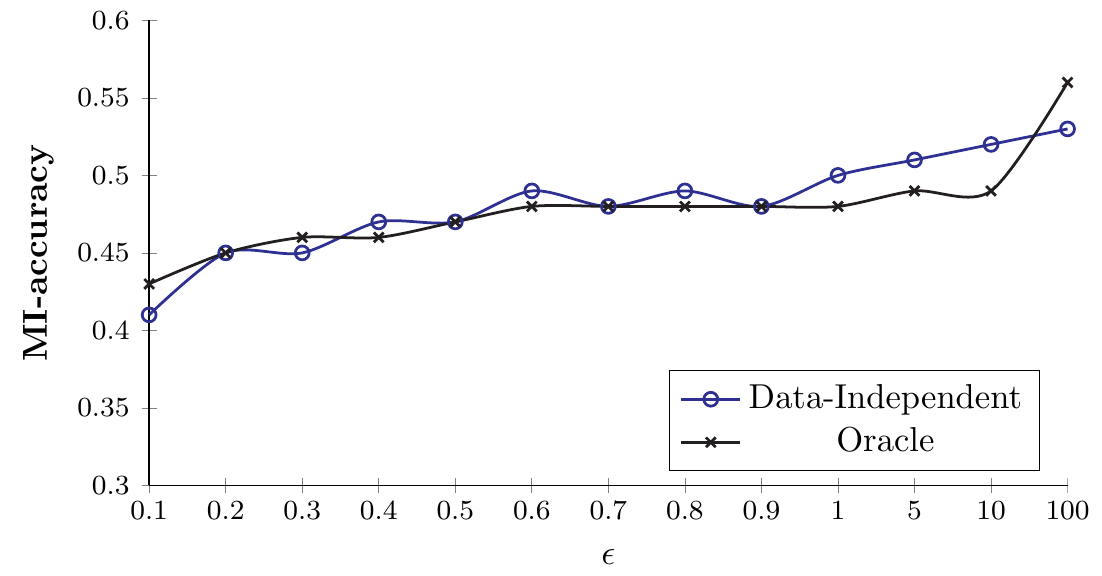}
  \vspace{-1\baselineskip}
  \caption{
    Model Invertibility:
    Data-Independent and Oracle Output Perturbation.
    The model inversion accuracy of these two algorithms are
    also very close with each other.
  }
  \label{fig:di-oracle-mi}
\end{figure}
\vskip 5pt
\cbstart
\noindent\textbf{Comparison with Other Private Mechanisms}.
For MI we also compare our method with projected-histogram algorithm
and objective perturbation. Specifically,
We notice that in the previous work of Fredrikson et al.~\cite{FLJLPR14},
they have demonstrated that projected-histogram algorithm actually
leaks more information and produces models with higher MI accuracy
than the functional mechanism (see the discussion under the head
``\textbf{Private Histograms vs. Linear Regression}''
in Section 4 of their paper).
For Objective Perturbation we find again that its MI accuracy is almost
the same as our Data-Independent and Oracle Output Perturbation.
This is not surprising because they have very close model accuracy.
\cbend

\vskip 5pt\noindent\textbf{For a fixed mechanism,
better utility gives more effective MI attacks.}
For a fixed mechanism, we demonstrate that
MI attacks get more effective as utility increases
with the availability of more training data.
\begin{figure*}[htp]
  \centering
  \subfloat[Model Accuracy]{
    \includegraphics[width=.4\columnwidth]{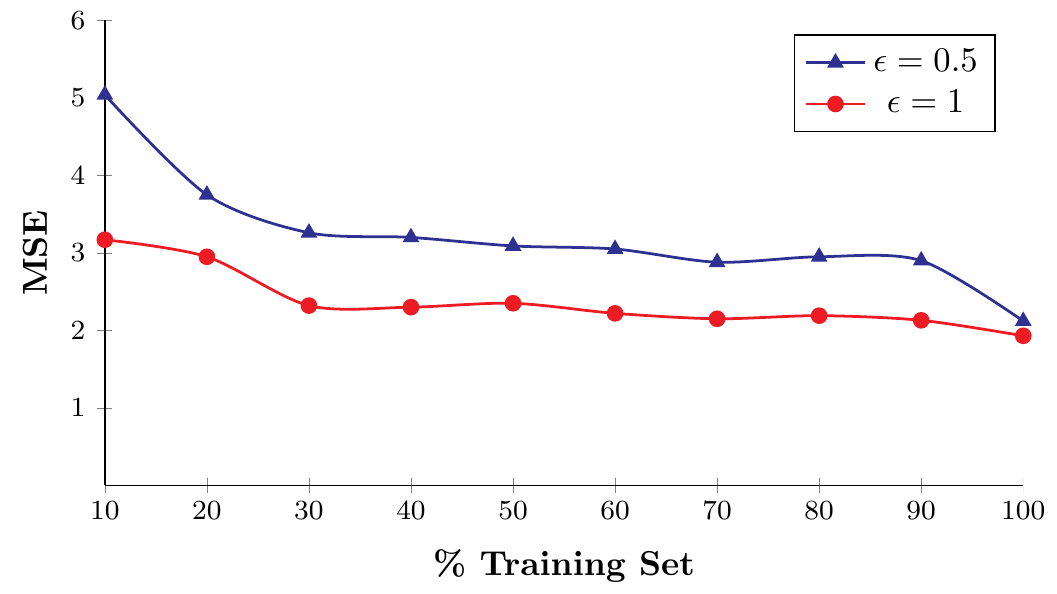}
  }
  \subfloat[Model Invertibility]{
    \includegraphics[width=.4\columnwidth]{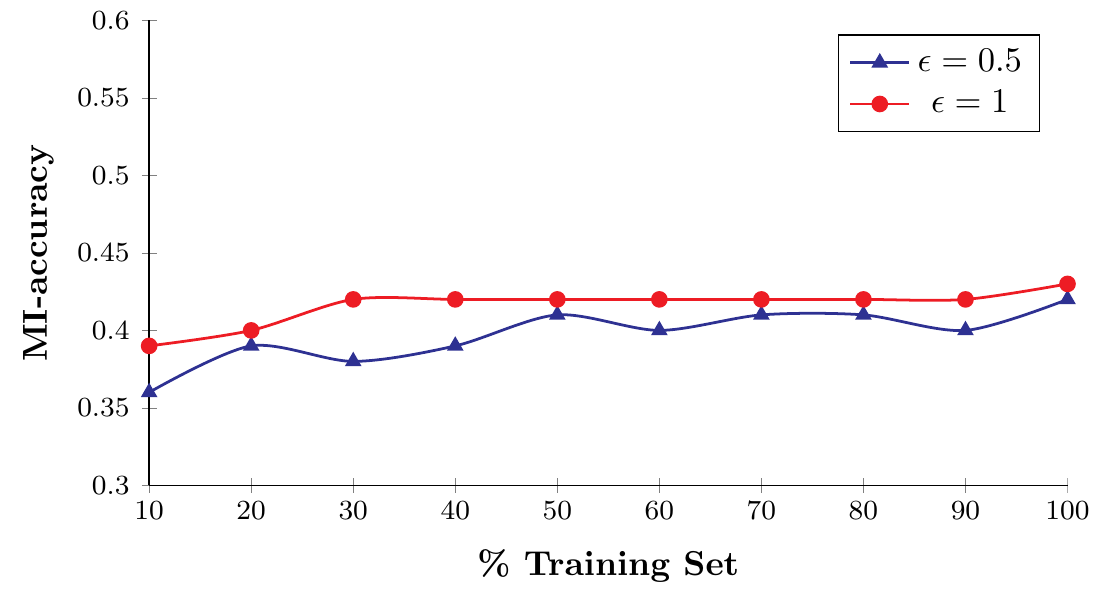}
  }
  \vspace{-1\baselineskip}
  \caption{Privately-Tuned Output Perturbation with increasingly larger training sets.
    The x-axis is the size of training chunks, and the y-axis is MSE.
    We use the following experimental method:
    The entire data set gets randomly permuted in the beginning.
    The data is split again into 25 chunks, with the first 24 for training
    and the last one for validation. For each of the training chunks,
    a fraction of the training items is sampled at rate $r$.
    We then train the model using the data with Tuned Output Perturbation
    and evaluate their utility and MI attack accuracy.
    This experiment is repeated $100$ time for $r=10\%, 20\%, \dots, 100\%$.}
  \label{fig:tuned-sizes-mse-mi}
\end{figure*}
Figure~\ref{fig:tuned-sizes-mse-mi} shows the results for Tuned Output Perturbation.
For $\varepsilon_p=0.5$, the mean square error drops from $5.05$
(with $r=10\%$ of available training data) to $2.9$ ($r=90\%$), while the MI
attack accuracy increases from $36\%$ ($r=10\%$) to $40\%$ ($r=90\%$).
\vskip 5pt
\noindent\textbf{A Theoretical Analysis}.
\noindent At first sight the above two phenomena may seem somewhat peculiar.
If MI attack is considered as a privacy concern,
then we have empirically observed that some privacy concern
becomes ``worse'', while differential privacy gets ``better.''

There is no contradiction here. Indeed,
it suffices to observe that DP is a property of the learning ``process'',
while MI attack is on the ``result'' of the process.
It is thus valid that the process satisfies a strong privacy guarantee,
while the result has some other concerns.
In the following, we give a ``lower bound'' result, which shows that,
as long as the optimal solution of the learning problem is susceptible to
MI attacks, improving DP/utility tradeoff will 
give effective MI attacks ``eventually.''

The intuition is as follows: Suppose that MI attack will be effective at
a hypothesis $w^*$ such that $L_{\cal D}(w^*) = L^*_{\cal D}$. That is,
the MI attack is effective at a hypothesis we want to converge to.
Now, suppose that the effectiveness of MI attack grows ``monotonically''
as we converge to $w^*$. Then, as long as the result of a learning algorithm
converges to $w^*$, it will gradually give more effective MI attacks.
We now give more details of this argument.

\vskip 5pt
\noindent\textbf{\em Assumptions}.
We make the following three assumptions on learning and MI attack:
(i) The utility of a hypothesis is measured by its generalization error.
(ii) Suppose that for $w^*$, $L_{\cal D}(w^*) = L_{\cal D}^*$,
and MI attack is effective for $w^*$.
(iii) As $|L_{\cal D}(w) - L_{\cal D}^*|$ gets smaller,
the MI attack for $w$ becomes more effective.

These assumptions are natural.
For (i), almost all previous work measures utility this way.
For (ii), since the ultimate goal of learning is to converge
to the best possible hypothesis, assuming MI attack will be effective
for such $w^*$ is natural. Finally, closeness in $L_{\cal D}$ indicates
that $w$ and $w^*$ are close in terms of their ``functioning as a model.''
Thus (iii) holds intuitively.
\vskip 5pt
\noindent\textbf{\em Better DP Mechanisms, More Effective MI attack}.
For any $n$, let $S_n$ denote a training set of size $n$.
Suppose that $A'$ is an $\varepsilon_p$-differentially private mechanism
with better privacy-utility tradeoff than the output perturbation mechanism $A$.
Thus with high probability for $w \sim A'(S_n)$,
$L_{\cal D}(w)$ is closer to $L_{\cal D}^*$ than that of $w$
sampled from $A(S_n)$. 
Combined with (ii) and (iii), we have that MI attack is more effective for
$w \sim A'(S_n)$.

One may note that the mechanism $A'$ could be any differentially private
mechanism, as long as it has better DP-utility tradeoff than \emph{output perturbation}.
For example, one can use the objective perturbation mechanism in~\cite{CMS11}
for generalized linear models, or exponential sampling based mechanisms
in~\cite{BST14} if minimizing training error is the goal.
\vskip 5pt
\noindent\textbf{\em For a Fixed Mechanism,
Better DP-Utility tradeoff gives More Effective MI attacks.}.
Theorem~\ref{thm:l2-stability-Lip} and our discussion at the end of
Section~\ref{sec:norm-stability} have indicated that as $n$ tends to the
infinity, the amount of noise injected for $\varepsilon_p$-differential privacy
vanishes to zero. Further, Theorem~\ref{thm:strongly-convex-dp-generalization}
and~\ref{thm:general-convex-dp-generalization} imply that for any convex
Lipschitz learning problem, the \emph{output perturbation} mechanism converges to
the optimal hypothesis $w^*$. Thus we are in the situation that as $n$ increases,
the output model has less noise yet it is closer to $w^*$. Therefore,
by assumptions (ii) and (iii) MI attack is more effective for larger $n$.
\vskip 5pt
\noindent\textbf{\em A Bayesian View Point}.
In a paper by Kasiviswanathan and Smith \cite{KS08},
the authors give a Bayesian interpretation of the semantics of differential
privacy. Informally speaking, given an arbitrary prior distribution over a
collection of databases, what differential privacy guarantees is that
\emph{two posteriors} obtained in two worlds of neighboring databases are
indistinguishable with each other.

What about the difference between \emph{prior and posterior}?
The same paper~\cite{KS08}, and indeed the original paper by
Dwork and Naor\cite{DN08}, have pointed out
that it is impossible to bound the difference between prior and posterior under
\emph{arbitrary background knowledge}.
Essentially, as long as the published information is ``useful,''
there exists some background knowledge that allows an adversary to learn and
significantly modify his/her prior.

Unfortunately, in MI attacks, some moderate background information
allows significant change of one's prior.
Therefore, while worsening MI attack is certainly not what one intended,
it is also of no surprise that better differentially
private mechanisms do not give better resilience against MI attack.



\section{Related Work}
\label{sec:related}
Differential privacy was proposed in the seminal work of Dwork, McSherry, Nissim and Smith~\cite{DMNS06}
and has become the de-facto standard for privacy.
Our setting -- learning a model differentially privately
-- was initiated in the work by Chaudhuri, Monteleoni and Sarwate~\cite{CMS11}.
Since the work of~\cite{CMS11},
a large body of work~\cite{BST14,CMS11,DJW13,JT13} has been devoted to this line,
culminating in a recent result by Bassily, Smith and Thakurta~\cite{BST14},
which obtains tight error bounds for general convex-Lipschitz learning problems and generalized linear models.
To this end, our work makes the underlying theme behind these works more explicit, namely,
the connection between differential privacy and stability theory in machine learning.

On the other hand, the application of these results in practice seems slow-paced.
The recent work by Fredrikson et al.~\cite{FLJLPR14} highlights two unfortunate facts on
training a differentially private pharmacogenetic model using the functional mechanism proposed by Zhang et al.~\cite{ZZXYW12}.
The first issue is the low model accuracy even for weak differential privacy guarantees.
For this problem, we show, both theoretically and empirically,
that the simple output perturbation can provide a much better DP-utility tradeoff
than the functional mechanism.

The second issue is about effectiveness of MI attacks.
MI attack is a new kind of privacy attack that was first described
in the same paper by Fredrikson et al.~\cite{FLJLPR14},
where the authors demonstrate that a DP mechanism can only prevent MI attack
with small privacy parameters.
We show that the privacy concern of MI attack is essentially orthogonal
to the concern of differential privacy.
To the best of our knowledge, this is the first work that establishes such a connection.


\section{Conclusion}
\label{sec:conclusion}
In this paper, we considered issues raised by Fredrikson et al.
about training differentially private regression models using functional mechanism.
Through an explicit connection between differential privacy and stable learning theory,
we gave a straightforward analysis showing that output perturbation
can be made to obtain, both theoretically and empirically,
substantially better privacy/utility tradeoff than the functional mechanism.
Since output mechanism is simple to implement,
this indicates that our method is potentially widely applicable in practice.
We went on to apply our theory to the same data set as used by Fredrikson et al.,
and the empirical results are encouraging.

We also studied model inversion attack,
a privacy attack raised in the same paper by Fredrikson et al.
We observed empirically that better differentially private mechanisms
lead to more effective model-inversion attacks.
We analyzed theoretically why this is the case.
We regard this result as a warning that care may be taken
when learning things from a data set,
even when strong differential privacy is guaranteed.


{\small
  \bibliographystyle{abbrv}
  \bibliography{paper}
}

\appendix
\section{Proofs}

\subsection{Proof of Proposition~\ref{prop:DP-implies-strongly-uniform-RO-stability}}
\vspace{1.5em}
\begin{proof}
  Let $f_S(w)$ be the probability density function of $A(S)$.
  Due to $\varepsilon$-differential privacy, then for any $S,i,z'$,
  $f_S(w) \le e^{\varepsilon}f_{S^{(i)}}(w)$.
  Therefore
  \begin{align*}
    &\left| \Exp[\ell(\tildeA(S), \bar{z})]
            - \Exp[\ell(\tildeA(S^{(i)}), \bar{z})]\right| \\
    =& \left| \int \ell(w,\bar{z})\Big(f_S(w) - f_{S^{(i)}}(w)\Big)dw \right| \\
    \le& \int |\ell(w,\bar{z})| |f_S(w) - f_{S^{(i)}}(w)|dw\\
    \le& B\int |f_S(w) - f_{S^{(i)}}(w)|dw \tag{1}
  \end{align*}
  Note that $(e^{-\varepsilon}-1)f_{S^{(i)}}(w)
  \le f_S(w) - f_{S^{(i)}}(w)
  \le (e^\varepsilon-1)f_{S^{(i)}}(w)$,
  so $|f_S(w) - f_{S^{(i)}}(w)|
      \le \max\{1-e^{-\varepsilon}, e^\varepsilon-1\} f_{S^{(i)}}(w)$.
  Plugging into (1) gives the first claimed inequality.
  The second inequality follows from the observation that
  $e^\varepsilon + e^{-\varepsilon} \ge 2$.
\end{proof}

\subsection{Proof of Lemma~\ref{lemma:exchanging}}
\vspace{0.5em}
\begin{proof}
  By the definition of $\vartheta_S$,
  \begin{align*}
    &\vartheta_S(u) - \vartheta_S(v) \\
    =&\ \Big(L_{S^{(i)}}(u) + \varrho(u) - \frac{1}{n}\ell(u, z') + \frac{1}{n}\ell(u, z_i)\Big) \\
    &\ \ \ \ \ \ - \Big(L_{S^{(i)}}(v) + \varrho(v) - \frac{1}{n}\ell(v, z') + \frac{1}{n}\ell(v, z_i)\Big) \\
    =&\ \vartheta_{S^{(i)}}(u) - \vartheta_{S^{(i)}}(v) 
        + \frac{\ell(v, z') - \ell(u, z')}{n}
        + \frac{\ell(u, z_i) - \ell(v, z_i)}{n}.
  \end{align*}
  Since $u$ minimizes $\vartheta_{S^{(i)}}$ so $\vartheta_{S^{(i)}}(u) - \vartheta_{S^{(i)}}(v) \le 0$, so
  \begin{align*}
    \vartheta_S(u) - \vartheta_S(v)
    \le \frac{\ell(v, z') - \ell(u, z')}{n} + \frac{\ell(u, z_i) - \ell(v, z_i)}{n}
  \end{align*}
  completing the proof.
\end{proof}

\subsection{Proof of Lemma~\ref{lemma:bound-norm-by-bounding-objective}}
\vspace{0.5em}
\begin{proof}
  We have that for any $\alpha \in (0,1)$,
  \begin{align*}
    \vartheta_S(v) &\le \vartheta_S(\alpha v + (1-\alpha)u) \\
    &\le \alpha\vartheta_S(v) + (1-\alpha)\vartheta_S(u) - \frac{\lambda}{2}\alpha(1-\alpha)\|v-u\|^2
  \end{align*}
  where the first inequality is because $v$ is the minimizer of $\vartheta_S$
  and the second inequality is by the definition of $\lambda$-strong convexity.
  By elementary algebra, this give that $\frac{\lambda}{2}\alpha\|v-u\|^2 \le \vartheta_S(u) - \vartheta_S(v)$.
  Tending $\alpha$ to $1$ gives the claim.
\end{proof}

\subsection{Proof of Theorem~\ref{thm:l2-stability-Lip}}
\vspace{0.5em}
\begin{proof}
  Let $u = A(S^{(i)})$, $v = A(S)$. By Lemma~\ref{lemma:bound-norm-by-bounding-objective},
  \begin{align*}
    \frac{\lambda}{2} \| u - v\|^2 \le \vartheta_S(u) - \vartheta_S(v) \tag{1}
  \end{align*}
  By Lemma~\ref{lemma:exchanging},
  \begin{align*}
    \vartheta_S(u) - \vartheta_S(v)
    \le \frac{\ell(v, z') - \ell(u, z')}{n} 
        + \frac{\ell(u, z_i) - \ell(v, z_i)}{n} \tag{2}
  \end{align*}
  Now because $\ell(\cdot, z)$ is $\rho$-Lipschitz, we have that
  $\ell(v, z') - \ell(u, z') \le \rho\|v - u\|$, and
  $\ell(u, z_i) - \ell(v, z_i) \le \rho\|u - v\|$.
  Plugging these two inequalities to (2), we have that
  $\vartheta_S(u) - \vartheta_S(v) \le \frac{2\rho}{n}\|u - v\|$.
  Plugging this to (1) and rearranging completes the proof.
\end{proof}

\subsection{Proof of Lemma~\ref{lemma:Lipschitz-condition-and-generalization-error}}
\vspace{0.5em}
\begin{proof}
  We have
  \begin{align*}
    | L_{\cal D}(w) - L_{\cal D}(A(S))|
    &= |\Exp_{z \sim {\cal D}}[\ell(w, z) - \ell(A(S), z)]| \\
    &\le \Exp_{z \sim {\cal D}}[|\ell(w, z) - \ell(A(S), z)|]
  \end{align*}
  For every $z \in Z$, $\ell(\cdot, z)$ is $\rho$-Lipschitz, so
  $|\ell(w, z) - \ell(A(S), z)| \le \rho\|w - A(S)\|_2$.
  Therefore
  \begin{align*}
    \Exp_{z \sim {\cal D}}[|\ell(w, z) - \ell(A(S), z)|] \le \rho\|w - A(S)\|_2.
  \end{align*}
  The proof is complete by observing that
  with probability at least $1-\gamma$ over $w \sim \widetilde{A}(S)$,
  $\|w - A(S)\|_2 \le \kappa(n,\gamma)$.
\end{proof}

\subsection{Proof of Theorem~\ref{thm:generalization-error-machinery}}
\vspace{0.5em}
\begin{proof}
  For any $S \sim {\cal D}^n$, we have that
  $L_{\cal D}(w) - L_{\cal D}^* 
   = \big(L_{\cal D}(w) - L_{\cal D}(A(S))\big)
     + \big(L_{\cal D}(A(S)) - L_{\cal D}^*\big)$.
  For $L_{\cal D}(A(S)) - L_{\cal D}^*$, we know that
  $\Pr_{S \sim {\cal D}^n}\big[ L_{\cal D}(A(S)) - L_{\cal D}^*
                                 > \varepsilon(n,\delta)
                           \big] < \delta$.
  Further, for every $S \sim {\cal D}^n$,
  from Lemma~\ref{lemma:Lipschitz-condition-and-generalization-error},
  $\Pr_{w \sim \widetilde{A}(S)}\big[ 
        L_{\cal D}(w) - L_{\cal D}(A(S))
         > \rho\kappa(n,\gamma) \big] < \gamma$.
  The proof is complete by a union bound.
\end{proof}

\subsection{Proof of Theorem~\ref{thm:strongly-convex-dp-generalization}}
\vspace{0.5em}
\begin{proof}
  Let $A$ denote the rule of empirical risk minimization,
  and $\widetilde{A}$ be its output-perturbation counter-part which ensures $\varepsilon_p$-differential privacy.
  Then by Theorem~\ref{thm:l2-stability-Lip},
  and corollary~\ref{cor:norm-noise-vector}, with probability at least $1-\gamma$ over $w \sim \widetilde{A}(S)$,
  $\|w - A(S)\|_2 \le \frac{4d\ln(d/\gamma)\rho}{\lambda n\varepsilon_p}$.

  Together with Theorem~\ref{thm:stochastic-strongly-convex-optimization},
  it follows that with probability at least $1-\delta'-\gamma$ over $S \sim {\cal D}^n$ and $w \sim \widetilde{A}(S)$,
  \begin{align*}
    L_{\cal D}(w) - L_{\cal D}^* \le \frac{4\rho^2}{\delta'\lambda n} + \frac{4d\ln(d/\gamma)\rho^2}{\lambda n\varepsilon_p}.
  \end{align*}
  Put $\delta' = \gamma = \delta/2$, thus with probability at least $1-\delta$,
  \begin{align*}
    L_{\cal D}(w) - L_{\cal D}^* \le \frac{4\rho^2}{\lambda n}\big(\frac{2}{\delta} + \frac{4d\ln(2d/\delta)}{\varepsilon_p}\big).
  \end{align*}
  Asymptotically this is $O\Big(\frac{\rho^2d\ln(d/\delta)}{\lambda n \delta\varepsilon_p}\Big)$.
  The proof is complete.
\end{proof}

\subsection{Proof of Theorem~\ref{thm:general-convex-dp-generalization}}
\vspace{0.5em}
\begin{proof}
  Let $\Bar{\ell}$ be defined as $\Bar{\ell}(w,z) = \ell(w,z) + \frac{\lambda}{2}\|w\|^2$.
  $\Bar{\ell}$ is $\lambda$-strongly convex and $(\rho + \lambda R)$-Lipschitz.
  Let $\Bar{L}_S$ and $\Bar{L}_{\cal D}$ be the empirical loss and true loss functions
  with respect to $\Bar{\ell}$. Note that for any $w \in {\cal H}$,
  $\Bar{L}_{\cal D}(w) = L_{\cal D}(w) + \frac{\lambda}{2}\|w\|^2$.

  By Theorem~\ref{thm:strongly-convex-dp-generalization},
  there is an $\varepsilon_p$-differentially private mechanism $\tildeA$ such that with probability
  $1-\delta$ over $S \sim {\cal D}^n$ and $w \sim \tildeA(S)$,
  \begin{align*}
    \Bar{L}_{\cal D}(w) - \Bar{L}_{\cal D}^*
    \le O\Big(\frac{(\rho + \lambda R)^2d\ln(d/\delta)}{\lambda n \delta\varepsilon_p}\Big) \tag{1}
  \end{align*}
  Let $\Bar{w} \in {\cal H}$ such that $\Bar{L}_{\cal D}(\Bar{w}) = \Bar{L}_{\cal D}^*$
  and $w^* \in {\cal H}$ such that $L_{\cal D}^* = L_{\cal D}(w^*)$.
  Because $\Bar{w}$ is the minimizer of $\Bar{L}_{\cal D}(\cdot)$, so
  \begin{align*}
    L_{\cal D}(w^*) + \frac{\lambda}{2}\|w^*\|^2 \ge L_{\cal D}(\Bar{w}) + \frac{\lambda}{2}\|\Bar{w}\|^2 \tag{2}
  \end{align*}
  Combining (1) and (2) we have
  \begin{align*}
    L_{\cal D}(w) - L_{\cal D}^* &\le O\Big(\frac{(\rho + \lambda R)^2d\ln(d/\delta)}{\lambda n \delta\varepsilon_p}\Big)
    + \frac{\lambda}{2}\left(\|w^*\|^2 - \|w\|^2\right) \\
    &\le O\Big(\frac{(\rho + \lambda R)^2d\ln(d/\delta)}{\lambda n \delta\varepsilon_p}\Big)
    + \frac{\lambda R^2}{2} \\
    &\le O\left(\frac{2\rho Rd\ln(d/\delta)}{n\delta\varepsilon_p} + \frac{\rho^2 d\ln(d/\delta)}{\lambda n\delta\varepsilon_p} + \lambda R^2\right)
  \end{align*}
  where the second inequality is because the hypothesis space is $R$-bounded.
  Putting $\lambda = \frac{\rho}{R}\sqrt{\frac{d\ln(d/\delta)}{n\delta\varepsilon_p}}$ gives the claimed bound.
\end{proof}

\subsection{Smooth Problems}
\label{appendix:smooth-convex-Lipschitz}

The following lemma bounds training error for smooth learning problems.
The lemma is first proved by Chaudhuri et al.~\cite{CMS11}
for generalized linear models.
\begin{lemma}
  \label{lemma:smooth-}
  Consider a learning problem $({\cal H}, Z, \ell)$.
  Suppose that $\ell$ is $\lambda$-strongly convex, $\rho$-Lipschitz,
  and $\beta$-smooth.
  Let $A$ be $\erm$ and $\tildeA$ be its output perturbation variant as described in
  Algorithm~\ref{alg:output-perturbation-strongly-convex}
  with privacy parameter $\varepsilon_p > 0$.
  Then for any training set $S \subseteq Z^n$,
  we have with probability at least $1-\gamma$ over $w \sim \tildeA(S)$,
  \begin{align*}
    L_S(w) - L_S(w^*) \le \beta\left(\frac{4d\ln(d/\gamma)\rho}{\lambda n \varepsilon_p}\right)^2.
  \end{align*}
  where $L_S(w^*) = \min_{w \in {\cal H}}L_S(w)$.
\end{lemma}
\begin{proof}
  Using Mean Value Theorem, we have that for some $\alpha \in (0,1)$,
  $L_S(w) - L_S(w^*) = \nabla L_S(\alpha w + (1-\alpha)w^*) (w - w^*)$,
  which is upper bounded by
  $\|\nabla L_S(\alpha w + (1-\alpha)w^*)\| \|w - w^*\|$ by
  the Cauchy-Schwarz inequality.

  Note that $w^*$ achieves the minimum, so $\nabla L_S(w^*) = 0$. Thus
  $\|\nabla L_S(\alpha w + (1-\alpha)w^*)\|
  = \| \nabla L_S(\alpha w + (1-\alpha)w^*) - \nabla L_S(w^*) \|$,
  which is upper bounded by $\beta \|\alpha(w-w^*)\| \le \beta \|w-w^*\|$
  because $\nabla L_S$ is $\beta$-Lipschitz.
  Therefore we have that $L_S(w) - L_S(w^*) \le \beta \|w-w^*\|^2$.

  Finally, note that with probability at least $1-\gamma$
  over $w \sim \tildeA(S)$,
  $\|w - w^*\| \le \frac{4d\ln(d/\gamma)\rho}{\lambda n \varepsilon_p}$, so
  $L_S(w) - L_S(w^*)
  \le \beta\left(\frac{4d\ln(d/\gamma)\rho}{\lambda n \varepsilon_p}\right)^2$.
\end{proof}


\end{document}